\documentclass[journal]{IEEEtran}

\usepackage{mathtools}
\usepackage[makeroom]{cancel}
\usepackage{amssymb}
\usepackage{graphicx}
\usepackage[english]{babel}
\usepackage{caption}
\usepackage{thmtools}
\usepackage{color}
\usepackage[font=footnotesize, margin=1cm]{caption}
\usepackage[framemethod=tikz]{mdframed}
\usepackage{lipsum}
\usepackage{comment}
\usepackage{bm}
\usepackage{comment}
\usepackage{tabularx}
\usepackage{changepage}   % for the adjustwidth environment
\usepackage{amsmath}
\usepackage[tight,footnotesize]{subfigure}

\DeclareMathOperator*{\argmax}{argmax}

\newcommand{\nn}{\nonumber}
\newcommand{\mc}{\mathcal}

\newcommand{\mbb}{\mathbb}

\newcommand{\beq}{\begin{equation}}
\newcommand{\eeq}{\end{equation}}

\IEEEoverridecommandlockouts

\usepackage{ifthen}
\newboolean{showcomments}
\setboolean{showcomments}{true}
\newcommand{\yue}[1]{\ifthenelse{\boolean{showcomments}}
{ \textcolor{red}{(Yue says:  #1)}}{}}

\newcommand{\BEAS}{\begin{eqnarray*}}
\newcommand{\EEAS}{\end{eqnarray*}}
\newcommand{\BEQ}{\begin{equation}}
\newcommand{\EEQ}{\end{equation}}
\newcommand{\BIT}{\begin{itemize}}
\newcommand{\EIT}{\end{itemize}}

\newcommand{\expec}{\mathbb{E}}

\newtheorem{theorem}{Theorem}
\newtheorem{lemma}{Lemma}

\newtheorem{corollary}{Corollary}

\newtheorem{RK}{Remark}

\usepackage{psfrag}

\usepackage{nomencl}

\makenomenclature

\usepackage{makeidx}
\makeindex

\usepackage{tikz}
\usetikzlibrary{shapes,backgrounds,calc}

\makeatletter
\tikzset{circle split part fill/.style  args={#1,#2}{%
 alias=tmp@name, % Jake's idea !!
  postaction={%
    insert path={
     \pgfextra{%
     \pgfpointdiff{\pgfpointanchor{\pgf@node@name}{center}}%
                  {\pgfpointanchor{\pgf@node@name}{east}}%
     \pgfmathsetmacro\insiderad{\pgf@x}
      \fill[#1] (\pgf@node@name.base) ([xshift=-\pgflinewidth]\pgf@node@name.east) arc
                          (0:180:\insiderad-\pgflinewidth)--cycle;
      \fill[#2] (\pgf@node@name.base) ([xshift=\pgflinewidth]\pgf@node@name.west)  arc
                           (180:360:\insiderad-\pgflinewidth)--cycle;            %  \end{scope}
         }}}}}
 \makeatother

\usepackage{pgffor}

\usepackage{algorithmic}
\usepackage{arydshln}

\renewcommand\baselinestretch{0.98} 

\begin{document}

\title{Indirect Mechanism Design for Efficient and Stable Renewable Energy Aggregation}

%Ex-post Stable and Fair Payoff Allocation for Renewable Energy Aggregation

\author{Hossein~Khazaei,~\IEEEmembership{Student Member, IEEE,}
        and~Yue~Zhao,~\IEEEmembership{Member, IEEE}
%\thanks{Some preliminary results from this work (cf. Section \ref{ExPostPropThePropMechan}) were presented in part at the IEEE PES Innovative Smart Grid Technologies Conference (ISGT), Washington, D.C., USA, 2017 \cite{ISGT2017}.}%
\thanks{H. Khazaei and Y. Zhao are with the Dept. of Electrical and Computer Engineering, Stony Brook University, Stony Brook, NY, 11794 USA (e-mails: \{hossein.khazaei, yue.zhao.2\}@stonybrook.edu).}%
}

\maketitle

\begin{abstract}
Mechanism design is studied for aggregating renewable power producers (RPPs) in a two-settlement power market. 
Employing an \emph{indirect} mechanism design framework, a payoff allocation mechanism (PAM) is derived from the competitive equilibrium (CE) of a specially formulated market with transferrable payoff. 
Given the designed mechanism, the strategic behaviors of the participating RPPs entail a non-cooperative game: 
It is proven that %the designed mechanism induces 
a unique pure Nash equilibrium (NE) exists among the RPPs, for which a closed-form expression is found. Moreover, it is proven that the designed mechanism achieves a number of key desirable properties at the NE: these include efficiency (i.e., an ideal ``Price of Anarchy'' of one), stability (i.e., ``in the core'' from a coalitional game theoretic perspective), and no collusion. In addition, it is shown that a set of desirable ``ex-post'' properties are also achieved by the designed mechanism. Extensive simulations are conducted and corroborate the theoretical results. 
\end{abstract}


\section{Introduction}

%renewables, issues, ways to address. 

Renewable energies play a central role in achieving a sustainable energy future. However, renewable energies such as wind and solar power are inherently non-dispatchable, and yet highly uncertain and variable. As a result, integrating renewable energies into power systems to serve loads raises significant reliability and efficiency challenges \cite{NERC2009Report, Pinson2013}. %SadeghianISGT2017
A variety of approaches have been proposed to compensate for the uncertainty of renewable energies, such as improving renewable power generation forecast \cite{Pinson2013}, employing better generation dispatch methods \cite{varaiya2011smart}, %Mohagheghi2016Part1, Mohagheghi2016Part2
energy storage deployment and control \cite{Bitar2011, Harsha14, CRZJG16}, %, CRZJG16, Harsha14, 
and demand response programs \cite{conejo2010real, li2011optimal, comden2017harnessing}. %, Sangrody2017, Foruzan2017

%aggregate, one key question is allocation, ex-ante work

Another solution that has received considerable attention is to aggregate statistically diverse renewable energy sources \cite{NERC2009Report, Baeyens2013, ZQRGP15}. 
In an aggregation, renewable power producers (RPPs) pool their generation together so as to reduce the aggregate uncertainty and risk, and hence the corresponding cost of compensation for their uncertainties. %in their integration into power systems. 
Accordingly, by forming an aggregation, RPPs can in total earn a higher payoff. A key question in aggregating RPPs is thus how to allocate the total payoff of an aggregation among its member RPPs. 

Notably, aggregating renewable energies has been studied extensively in the context of a two-settlement power market model, consisting of a forward power market and a real time one. As such, RPPs participate in these markets in the same way as conventional generators do. %An RPP that participates in these markets would sell a forward power contract certain time (e.g., a day) ahead of its actual delivery. At the delivery time, any deviation from the forward contract would need to be cleared from buying / selling power in the real time market. %as this model captures the essence of the issues with renewable energy uncertainty. 
With this model, allocating payoffs in an aggregation of RPPs has been studied in a \emph{coalitional game} framework based on the \emph{joint probability distribution} of all the RPPs' uncertain generation \cite{Baeyens2013, ZQRGP15}. The primary interest in this setting is to find a payoff allocation solution that is \emph{stable/in the core} of the game. This is in general computationally hard in the sense that the number of constraints of the corresponding optimization problem grows exponentially with the number of RPPs. To this end, the core is proven to be non-empty in \cite{Baeyens2013}, %(i.e., such payoff allocation exists), 
and a closed-form solution of a payoff allocation in the core is found in \cite{ZQRGP15}. 

%aggregator's issue, allocation depends on information, a necessary step, info collection, MD, %game. 

%However, the computation of a stable/in the core payoff allocation in an aggregation of RPPs critically depends on the \emph{knowledge} of the joint 

While this line of works achieve efficiency (with an optimal forward contract) and stability (with a payoff allocation in the core) in aggregating RPPs, an underlying assumption is that the aggregator \emph{knows the joint probability distribution of the RPPs' generation}. %Arguably, this assumption may not be well scalable as the number of RPPs increases, as each RPP may have their private information that better characterized
In practice, however, an aggregator typically does not have the best or full knowledge of such information about the RPPs: not only the amount of relevant information can be overwhelming to glean, but also the RPPs themselves often have better information privately about their own generation than the aggregator does. As a result, to aggregate renewables in practice, 
it is essential to consider an \emph{information collection} step by the aggregator with the RPPs. %that precedes its participation in the markets as well as payoff allocation. 
Consequently, with this step, aggregating renewable energies in a two-settlement market constitutes a \emph{mechanism design} problem (as will be shown in details in Section \ref{sec:MD}). In short, the primary goal of such mechanism design is the following: \emph{Granted that all the RPPs behave strategically for their own interests under this mechanism}, the aggregation can still achieve the same desirable outcome \emph{as if all the RPPs' information are indeed known to the aggregator}. 

% // in general, info can be anything, 
%when not all, it's indirect, we as little as possible, the framework, prior work, 

In general, there's complete freedom in designing the information collection step of the mechanism. In particular, when the aggregator does \emph{not} elicit {all available information} from the RPPs, the mechanism is called an \emph{``indirect''} one , (in contrast to a ``direct'' one when all information from the RPPs are requested by the aggregator upfront) \cite{AlgorithmGameTheory}. Indeed, we would like an aggregator to elicit \emph{as little information as possible} from the RPPs, while still guaranteeing the performance of the overall mechanism. To this end, 
%Relaxing this assumption, another line of works have studied payoff allocation solely based on \emph{realized} generation of RPPs. Notably, 
a simple interface between aggregator and RPPs has been proposed in \cite{Nayyar13}: each RPP submits \emph{just a single number} to the aggregator, and the aggregator simply passes on the sum of these numbers as the forward power contract for the entire aggregation. %With  simple design of the information collection and commitment steps for the aggregator, 
Based on this simple interface, 
the central design task is again on the \emph{payoff allocation} among the RPPs, for which a number of payoff allocation mechanisms (PAMs) have been proposed \cite{Nayyar13, LinBitar, Harirchi14, Harirchi16, ISGT2017, chakraborty2017cost}. 
Under any given PAM, the RPPs' strategic decision making %at the time of forward power market 
entails a \emph{non-cooperative game} (as will be described later in Section \ref{noncoop}), and properties of the Nash Equilibria of this game have been studied in \cite{Nayyar13, LinBitar, Harirchi14, Harirchi16}. The existing PAMs in the literature, however, have only gained limited success, as some essential and highly desired properties still cannot be achieved. In particular, achieving efficiency and stability/in the core at the Nash Equilibria remains to be an open question. 
Lastly, we note that mechanism design methods have also been employed in power markets for problems other than renewable energy aggregation, e.g., for incentivizing conventional generators to reveal truthful information \cite{silva2001application}. 
%\footnote{To be precise, efficiency and stability need to be achieved at equilibria in a non-cooperative game of RPPs, as will be described in Section \ref{noncoop}.} 

In this paper, we investigate indirect mechanism design under the framework of the above simple interface. We propose a new payoff allocation mechanism, and show that all the essential desirable properties are achieved by this PAM. We first show that, given the designed mechanism, the outcome of the mechanism can be predicted by a \emph{unique Nash equilibrium} (NE) among the RPPs, for which we provide a \emph{closed-form} expression. Moreover, this unique NE is \emph{efficient}, meaning that it achieves the \emph{maximum} total payoff \emph{as if all information are known a-priori to the aggregator}. Next, we show that, the proposed payoff allocation is \emph{stable/in the core} at this unique NE, meaning that \emph{no subset} of the RPPs have any incentive to leave the aggregation as they cannot possibly earn a higher payoff on their own. Furthermore, we show that the designed mechanism guarantees \emph{no collusion} among the RPPs at the unique NE, as they cannot earn a higher payoff by colluding. {Lastly, we show that a set of \emph{ex-post} properties are achieved by the proposed mechanism (with results reported in part in \cite{ISGT2017}). We note that similar ex-post results have also been independently developed in a recent work \cite{chakraborty2017cost}.} 

The remainder of the paper is organized as follows. The problem is formulated in Section \ref{ProblemFormul}, and the indirect mechanism design framework for aggregating RPPs is introduced. %two-settlement market model, the framework of aggregating and allocating payoffs to the RPPs, and the induced non-cooperative game among the RPPs. 
The design goals, i.e., the desired properties of the PAM %Payoff Allocation Mechanism %(PAM) 
are described in Section \ref{DesiredPropert}. The main results are presented in Section \ref{mainresults}, in which we show that the proposed PAM achieves all the desired properties. Analysis and proofs of the main results are provided in Section \ref{sec:ana}. Another set of properties achieved in an ``ex-post'' sense by the proposed PAM are presented in Section \ref{ExPostPropThePropMechan}. Simulations are conducted in Section \ref{Simul}. Conclusions are drawn in Section \ref{Concl}.

%			
%
%~
%
%~
%
%~
%
%~
%
%~
%
\section{Problem Formulation}   \label{ProblemFormul}
%\subsection{How The Market Works}
%\subsubsection{Renewable Power Producers in a Two Settlement Power Market} \label{sec:twoset}
\subsection{System Model} \label{sec:twoset} %Renewable Power Producers in a Two Settlement Power Market

We consider RPPs participating in a two-settlement power market consisting of a day-ahead (DA) and a real time (RT) market. 
As a baseline case, we first consider an RPP $i$ who participates in the market \emph{separately} from the other RPPs. 

In the DA market, RPP $i$'s generation at the time of interest in the next day is modeled as a random variable, denoted by $X_i$. 
(We assume that the joint probability density function for the vector of random variables $X_1,\ldots,X_N$ exists.) 
RPP $i$ then  determines a forward power supply contract in the amount of $c_i$ to sell in the DA-market. Interchangabely, $c_i$ is also termed a day-ahead commitment. 
RPP $i$ gets a payoff of $p^f c_i$ where $p^f$ denotes the price in the DA market. 

At the delivery time in the next day, 
RPP $i$ obtains its realized generation $x_i$: a) If it faces a shortfall, i.e., $c_i - x_i>0$, it needs to purchase the remaining power from the RT market at a real-time buying price $p^{r,b}$, b) if it has excess power, i.e. $x_i - c_i>0$, it can sell it in the RT-market at a real-time selling price $p^{r,s}$. In case excess power needs to be penalized as opposed to rewarded, we model such cases by having $p^{r,s}<0$.  
%Unless otherwise stated, 
We make the assumption that $p^{r,s} \le p^{r,b}$, which must hold for no arbitrage. %\yue{More general model on prices, say that RT prices are expectations}
%Because it is (in expectation) more expensive to buy RT power for a shortfall than to sell for a surplus, 
Intuitively, the higher uncertainty an RPP's generation has at DA, the more cost it incurs to the RPP. 

Specifically, the \emph{realized} payoff of an RPP $i$ who separately participates in the market is given by 
\begin{align}   \label{PayoffSep}
\mc{P}_i^{sep}  \triangleq p^f c_i - p^{r,b} \left(c_i - x_i\right)_+ + p^{r,s} \left(x_i - c_i\right)_+ 
\end{align}
where 
$(\cdot)_+ \triangleq \max(0,\cdot)$\footnote{We use the symbol $\triangleq$ to define notations.}. We denote the expected payoff of RPP $i$ at the time when $c_i$ is determined one day ahead by 
\begin{align} \label{pisep}
\pi_i^{sep}(c_i) \triangleq \expec[\mc{P}_i^{sep}],
\end{align} 
where the expectation is taken over the random generation $X_i$. 
\begin{RK}[Model Assumptions on Prices]
In this paper, we consider a price taking environment for the RPPs in the DA market, and thus the price $p^f$ is given. We also consider that some fixed values for the RT buying and selling prices $p^{r,b}$ and $p^{r,s}$ are assumed at DA. These values can be interpreted as the RPPs' expectations of the RT prices if they were to experience a shortfall or a surplus, respectively. We note that the price taking assumption is a simplifying one, which assumes that none of the RPPs can affect the price significantly at DA due to its relatively small size. %While the price taking assumption is a simplifying one, 
This provides a first order approximation of the problem that allows us to perform effective analysis and gain insight. The results will lay the foundation for further investigation of more general scenarios. 
%Indeed, 
In particular, 
we would like to note that our latest result following this paper has achieved some success in relaxing these assumptions, and has addressed the price-making scenarios in both DA and RT markets \cite{HZ18}. 
\end{RK}

\subsection{Aggregating Renewable Energies} 
%\subsubsection{Preliminaries: Stable Payoff Allocation in An Aggregation of RPPs} 
%\subsection{Payoff Allocation %for aggregating renewables 
%with Full Information of RPPs}  %Ideal Scenario: Stable Payoff Allocation in An Aggregation of RPPs Where All Information Is Public 
We consider an aggregator that aggregates the power generation from a set of $N$ RPPs, denoted by $\mc{N}$, and participates in the DA-RT market on behalf of the RPPs. Intuitively, aggregation reduces the total uncertainty for the RPPs due to the statistical compensation %diversity 
among the random power generation at DA, and hence brings economic benefit to them. In this paper, transmission network constraints are not considered, and are left for future work.  

In general, the aggregator takes actions in the DA and RT markets as follows:

\begin{enumerate}
	\item[a.] In the DA market, the aggregator determines an amount of forward power contract to sell, denoted by $c_{\mc{N}}$. %= \sum_{i \in \mathcal{N}}{c_i}$. 
	\item[b.] %The reason for using this particular model is detailed in \cite{Zhao2016}. 
	At the delivery time, the aggregator collects all the realized generation from the RPPs, denoted by $x_{\mc{N}} = \sum_{i \in \mathcal{N}}{x_i}$, to meet the commitment $c_{\mc{N}}$. The deviation is settled in the RT market in the same way as in Section \ref{sec:twoset}. The realized payoff of the aggregator is thus %given by
	\begin{align}\label{PayoffAgg}
	\hspace{-2pt} \mc{P}_\mc{N}  \hspace{-1pt} \triangleq p^f c_{\mc{N}} \hspace{-2pt}  -   p^{r,b} \hspace{-2pt} \left(c_{\mc{N}} \hspace{-2pt} - \hspace{-1pt} x_{\mc{N}}\right)_+ \hspace{-2pt} + p^{r,s} \hspace{-2pt} \left(x_{\mc{N}} \hspace{-2pt} - \hspace{-2pt} c_{\mc{N}}\right)_+
	\end{align}
	Next, the aggregator returns a payoff $\mc{P}_i $ to each RPP $i$. 
\end{enumerate}
\begin{RK}[Budget Balance]
In this paper, we require the aggregator's budget balance be satisfied \emph{in all circumstances}: 
\begin{align}
\sum_{i=1}^N \mc{P}_i = \mc{P}_\mc{N}. 
\end{align}
We note that this is a stronger condition than just requiring budget balance be satisfied in expectation. 
\end{RK}

Accordingly, there are two decisions an aggregator needs to make: a) the total commitment $c_\mc{N}$ at DA, and b) the set of payoff allocations $\{\mc{P}_i\}$ at RT. 
In making these decisions, two fundamental goals an aggregator would like to achieve are:
\begin{itemize}
\item \emph{Efficiency}: The \emph{total} expected payoff of the aggregation $\pi_\mc{N} = \expec[\mc{P}_\mc{N}]$ is maximized. 
\item \emph{Fairness / Stability}: The payoff allocation within the aggregation $\{\mc{P}_i\}$ are fair to each RPP. In this paper, we interpret fairness using the notion of \emph{stability/in the core} from a coalitional game perspective, as will be described in detail in Section \ref{DesiredPropert}. 
\end{itemize} %a) ``efficiency'', that is  maximize the \emph{total} expected payoff of the aggregation $\pi_\mc{N} = \expec[\mc{P}_\mc{N}]$, 
%In addition, a budget balance requirement $\mc{P}_\mc{N} = \sum_{i=1}^N \mc{P}_i$ is also desired. Note that such a budget balance requirement needs to hold at all times, not merely in expectation. 
In particular, a) achieving efficiency depends on the aggregator's decision on the total DA commitment $c_\mc{N}$, and b) with the optimal $c_\mc{N}$, achieving stability depends on the decisions on the RT payoff allocation $\{\mc{P}_i\}$. 

\subsubsection*{The Ideal Case of Aggregator Having Full Information}
To achieve efficiency and stability, making decisions on $c_\mc{N}$ and $\{\mc{P}_i\}$ requires the aggregator to know sufficient information from the RPPs. The ideal case would be an aggregator with \emph{full information} from the RPPs, in particular, the DA \emph{joint probability distribution} of all the random generation $\{X_i\}$. %is a key information based on which $c_\mc{N}$ and $\{\mc{P}_i\}$ shall be decided. 
Based on the joint probability distribution of $\{X_i\}$, closed-form solutions of $c_\mc{N}$ and $\{\mc{P}_i\}$ that achieve efficiency and stability have been found in \cite{ZQRGP15}, and will be used for numerical comparisons later. 

%While this ideal case s solved, 

\subsection{The Mechanism Design Problem} \label{sec:MD}
%\subsubsection{Mechanism Design %for aggregating renewables 
%with No Prior Information}

In practice, however, 
it is unlikely for an aggregator to precisely know the DA joint probability distribution of $\{X_i\}$ for a number of reasons: a) the best information on future power generation may only be \emph{privately} known to the RPPs, and b) the amount of information can be overwhelmingly large and difficult to glean for a single aggregator, especially when the number of RPPs becomes large, (consider, e.g., hundreds of thousands of distributed energy resources in a power distribution system). 

\emph{In this paper, we do not assume the aggregator knows any information a-priori at DA on the RPPs' random generation $\{X_i\}$.} Instead, we consider a general framework in which the aggregator \emph{elicits information} from the RPPs, based on which decisions on $c_\mc{N}$ and $\{\mc{P}_i\}$ are then made. As such, the aggregator's actions involve the following three general steps: 
\begin{enumerate}
\item[a.] \emph{Information Collection:} At DA, the aggregator elicits certain information from the RPPs. 
\item[b.] \emph{Commitment:} At DA, the aggregator determines a total DA commitment $c_\mc{N}$. 
\item[c.] \emph{Payoff Allocation:} At RT, the aggregator allocates a payoff $\mc{P}_i$ to RPP $i, ~i=1,\ldots,N$. 
\end{enumerate}
Specifying how these three steps are performed constitutes a \emph{mechanism design} problem. It is important to note the \emph{generality} of this design problem, as there is complete freedom in choosing what information to request from the RPPs, how they are used to determine $c_\mc{N}$, and how payoffs are allocated. 

\subsection{An Indirect Mechanism Design Framework} \label{indmec}
In this mechanism design problem, it is \emph{not} imperative for the aggregator to elicit \emph{all} information from the RPPs. When the aggregator does elicit all information upfront, such a mechanism is called a ``direct mechanism''; %As discussed above, collecting all information can be an overwhelming task for an aggregator 
Otherwise, it is called an ``indirect mechanism'' \cite{AlgorithmGameTheory}. 
Rather than restricting ourselves to direct mechanisms, more generally, we will investigate \emph{indirect} mechanism design: We would like to \emph{elicit as little information from the RPPs as possible, while still guaranteeing efficiency and stability of the aggregation}.  

In particular, we will investigate the following framework of indirect mechanisms employing a simple design of Steps a. (Information Collection) and b. (Commitment) \cite{Nayyar13}. 
\begin{enumerate}
\item[a.] At DA, the aggregator elicits \emph{a single number} $c_i$ from each RPP $i$. 
\item[b.] At DA, the aggregator commits $c_\mc{N} = \sum_{i=1}^N c_i$. 
\item[c.] At RT, the aggregator allocates a payoff $\mc{P}_i$ to RPP $i, ~i=1,\ldots,N$. 
\end{enumerate}
We term the number $c_i$ submitted by RPP $i$ its \emph{DA commitment}. Accordingly, the aggregator simply passes on the \emph{sum of the RPPs' DA commitments as the aggregate commitment}. 
As Steps a. and b. are now fully specified, the central design task is \emph{Step c. -- Payoff Allocation}. As such, the design and analysis of \emph{Payoff Allocation Mechanisms} (PAMs) that determines $\mc{P}_i$ would be the focus of the remainder of the paper. %We term  that determines $\mc{P}_i$ \emph{Payoff Allocation Mechanism}. 

\subsection{Mechanism's Outcome in a Non-Cooperative Game of Strategic RPPs} \label{noncoop}
Given any PAM that specifies the rule of determining $\mc{P}_i$, a key question is how one predicts the \emph{outcome} under this mechanism. To answer this question, we must understand the behavior of the RPPs given any mechanism. As an RPP $i$ is free to submit any DA commitment $c_i$ to the aggregator, a \emph{rational and strategic} RPP $i$ would submit a $c_i$ at DA that \emph{maximizes} its expected allocated payoff $\expec[\mc{P}_i]$. It is important to note that, given a PAM, $\expec[\mc{P}_i]$ can also depend on the \emph{other} RPPs' submissions of commitments.
%\begin{comment}
%Accordingly, we denote  
%$\expec[\mc{P}_i]$ by $\pi_i \left(c_i , \left\{c_{-i}\right\}\right)$, where $\left\{c_{-i}\right\}$ denotes the set of commitments of the RPPs other than $i$. Note that $\pi_i \left(c_i , \left\{c_{-i}\right\}\right)$ depends on the particular design of PAM. 
%\end{comment}
{
Accordingly, we denote the expected payoff of RPP $i$ by
\begin{align}
\pi_i \left(c_i , \left\{c_{-i}\right\}\right) \triangleq \expec[\mc{P}_i],
\end{align}
where $\left\{c_{-i}\right\}$ denotes the set of the commitments of the RPPs other than $i$. Note that $\pi_i \left(c_i , \left\{c_{-i}\right\}\right)$ depends on the particular design of the PAM. 
}

Therefore, the strategic decision making of the $N$ RPPs on their submissions $\{c_i\}$ at DA can be studied under a \emph{non-cooperative game} framework, (termed a ``contract game'' in \cite{Nayyar13}.) To predict the \emph{outcome} of any designed mechanism, a natural solution concept is the \emph{Nash equilibria} of this non-cooperative game\footnote{We note that, in practice, NE may not be achieved in a dynamic market. Investigation of other solution concepts is left for future work.}. %: A natural prediction is the \emph{Nash equilibria}. 
Specifically, 
given a PAM, a set of commitments $\{c_i^{ne}\}$ is at a pure Nash equilibrium (NE) if they satisfy 
\begin{align} \label{defNE}
c_i^{ne} \in \argmax_{c_i} \pi_i \left(c_i , \left\{c_{-i}^{ne}\right\}\right), ~\forall i. %\forall i=1,\ldots,N, ~~~
\end{align}
As such, an NE offers a stable\footnote{This notion of stability from NE is \emph{not} the coalitional game theoretic stability %as a design goal for renewable energy aggregation, 
which will be described as Property 5) in the next section.} outcome of the RPPs' decision making on their commitments, as no RPP has any incentive to deviate from its already best responding commitment. 
In the remainder of the paper, we devote the notation $\{c_i^{ne}\}$ to denoting a set of commitments at a pure NE.

As a mechanism designer for aggregating RPPs, we are interested in designing a PAM so that a set of essential and desirable properties can be achieved \emph{at equilibria} of this non-cooperative game, given the designed PAM.

\section{Design Goals: \\ Desired Efficiency and Stability Properties}  \label{DesiredPropert}

In this section, we provide a 
%We begin with a 
detailed description of the desired properties in designing PAM for aggregating RPPs. %In short, these properties Achieving efficiency and stability at the outcome induced by the mechanism in the game of RPPs . %We then show that how the proposed PAM achieves all these properties. 

%This PAM achieves a range of ex-ante and ex-post PAMs: 
%\subsection{Ex-Ante Properties} \label{ExAntePro}
%We first describe the desired properties of the expected payoffs \emph{when RPPs submit $\{c_i\}$ one day ahead of delivery}. We refer to these as ``ex-ante'' properties. 

\begin{enumerate} 
	\item \emph{Existence and uniqueness of pure Nash equilibrium}:
	%Given a PAM, an NE is a set of commitments $\{c_i^\star\}$ that satisfy $\forall i, c_i^\star \in \argmax_{c_i} \pi_i \left(c_i , \left\{c_{-i}^\star\right\}\right)$. 
	For a mechanism to have predictable outcomes, it is desired that the non-cooperative game among the RPPs induced by the mechanism has a unique pure NE, %It is even more desirable if there is only one pure NE, 
	which would be the unique outcome that one shall predict from the RPPs' strategic decision making. 
	
	\item \emph{Efficient computation of NE}: In particular, we are interested in whether the unique pure NE, if exists, can be computed in \emph{closed form}. % expression. with closed form formula to compute itThe contract game imposed by the proposed PAM \eqref{NewPAM} posseses a unique pure Nash equilibrium (NE), and has the following closed form formula to compute it: for $i = 1, \dots, N$
	%
	%\item \emph{Closed form solution of the pure NE}: The unique pure NE of the contract game has a closed form formula to compute it: for $i = 1, \dots, N$	

	%
	\item \emph{Efficiency}: A PAM is efficient if, at the NE, the aggregation achieves the \emph{maximum} expected payoff for the entire group of RPPs. Specifically, this means that $c_\mc{N}^{ne} \triangleq \sum_{i=1}^N c_i^{ne}$ is equal to 
	the optimal commitment for the entire aggregation 
	\begin{align}
	c_{\mc{N}}^{\star} \triangleq \argmax_{c_\mc{N}} \expec[\mc{P}_\mc{N}]
	\end{align} (cf. \eqref{PayoffAgg}). % of the newsvendor problem for the aggregator \cite{Bitar2012} and maximizes the expected payoff of the aggregator.
	This optimal commitment can in fact be computed as a solution to a news-vendor problem (for which we refer the readers to \cite{Bitar2012} for more details): $c_{\mc{N}}^{\star} = \mc{F}_{\mc{N}}^{-1} \left(\frac{p^f - p^{r,s}}{p^{r,b} - p^{r,s}}\right)$, where  
$\mc{F}_{\mc{N}} \left(x_{\mc{N}}\right)$ is the cumulative distribution function (cdf) of the \emph{aggregate} random generation $X_\mc{N} = \sum_{i=1}^N X_i$. In this paper, we assume that the inverse function $\mc{F}_{\mc{N}}^{-1}(\cdot)$ exists. % We use $\left(c_i^{\star}, \left\{c_{-i}

	\item \emph{Individual rationality}: At the NE, the expected payoff of RPP $i$ should be at least as high as the \emph{maximum} payoff it could have gotten had it separately participated in the DA-RT market. Specifically, %we need 
	\begin{align} \label{IR}
	\forall i \in \mathcal{N}, ~~~\pi_i \left(c_i^{ne}, \left\{c_{-i}^{ne}\right\}\right)  \geq \pi_i^{sep} \left( c_i^{\star , sep} \right),  
	\end{align}
	where %$c_i^{\star , sep}$ %is the \emph{optimal} DA commitments of RPP $i$ had it separately participated in the DA-RT market. In other words,  
	\vspace{-5pt}
\begin{align}\label{optsep}
c_i^{\star , sep} \triangleq \argmax_{c_i} \pi_i^{sep} \left(c_i\right)
\end{align} (cf. \eqref{PayoffSep} and \eqref{pisep}). 
It is important to note that, the (separately) optimal commitment $c_i^{\star , sep}$ is in general \emph{not equal} to the equilibrium commitment $c_i^{ne}$.  %$\pi_i^{sep} \left( c_i^{\star , sep} \right) = \max_{c_i} \pi_i^{sep} \left(c_i\right)$ and $c_i^{\star , sep} = \argmax_{c_i} \pi_i^{sep} \left(c_i\right)$. 
With individual rationality satisfied \eqref{IR}, not a single RPP has any incentive to leave the aggregation. 
	\item \emph{Stability / In the core}: A generalization of individual rationality to a much stronger sense of stability is being ``in the core'', a property celebrated in coalitional game theory \cite{GameBook}. Specifically, being in the core %in an ``ex-ante'' sense 
	means that the RPPs' expected payoffs satisfy the following condition: 
	if any subset $\mc{T}$ of the RPPs leave the aggregation, separately form their own aggregation, and then participate in the market based on their aggregate generation $X_\mc{T} \triangleq  \sum_{i\in\mc{T}} X_i$, their highest possible expected payoff would be %(i.e. they use the $c_{\mathcal{T}}^{\star , sep}$), 
%they will get an expected payoff 
no higher than the sum of their expected payoffs originally from the PAM at the NE. Specifically, 
	\begin{align}  \label{ExAnteStronglyInCore}
	\forall \mc{T} \subset \mc{N}, \  \sum_{i \in {\mc{T}}} \pi_i \left(c_i^{ne} , \{c_{-i}^{ne}\} \right) \geq  \pi_{\mc{T}}^{sep} \left(c_{\mc{T}}^{\star , sep}\right), 
	\end{align}
	where 
	\vspace{-5pt}
	\begin{align} 
	\pi_{\mc{T}}^{sep}(c_{\mc{T}}) \triangleq \expec \big[ & p^f c_\mc{T} - p^{r,b} \left(c_\mc{T} - X_\mc{T}\right)_+  \nn\\
	 & + p^{r,s} \left(X_\mc{T} - c_\mc{T}\right)_+ \big], 
%	\pi_{\mc{T}}^{sep}(c) \triangleq \expec \big[ & p^f c - p^{r,b} \left(c - X_\mc{T}\right)_+ \!+ p^{r,s} \left(X_\mc{T} - c\right)_+ \big], 
	\end{align}
	and ~$c_{\mc{T}}^{\star , sep} \triangleq \argmax_{c_{\mc{T}}} \pi_{\mc{T}}^{sep}(c_{\mc{T}})$. 
%	\begin{align} 
%	c_{\mc{T}}^{\star , sep} = \argmax_{c_{\mc{T}}} &~~\pi_{\mc{T}}^{sep}(c_{\mc{T}}) \\
%	 = \argmax_{c_{\mc{T}}} &~~p^f c_\mc{T} - p^{r,b} \left(c_\mc{T} - x_\mc{T}\right)_+ \\ 
%	 &~~+ p^{r,s} \left(x_\mc{T} - c_\mc{T}\right)_+. 
%	\end{align}
As such, being in the core means that, not only every single RPP, but \emph{all subsets} of RPPs do not have any incentive to leave the aggregation. This implies a very strong sense of stability of an aggregation. % of an aggregation. 
	
%	As a result, being in the core implies that the grand coalition/aggregation is stable in the ex-ante sense.
	
	\item \emph{No collusion}: Suppose a subset of RPPs %, denoted by $\mc{S}$, 
	join together as a single player \emph{before} participating in the aggregation with the remaining RPPs. 
	{For now, we assume the remaining RPPs know the joining of these RPPs as a single player. Later, we will show that the DA commitments of the remaining RPPs at the NE actually do \emph{not} depend on whether or not they know there is a collusion.} 
	Because of the change of the set of players, %in the game setting, 
	a new game, and hence new NE would arise. 
	The expected payoff of this ``joint player'' at the new NE should be \emph{no higher} than the sum of these RPPs' expected payoffs at the NE of the original game. %where they do not form any pre-coalition beofore joining the main coalition.%, i.e. $ \sum_{ i \in \mc{T}} \pi_i \left( c_i^{\star} , \left\{ c_{-i}^{\star} \right\} \right) \geq \pi_{\mc{T}} \left( c_{\mc{T}}^{\star} , \left\{ c_{-\mc{T}}^{\star} \right\} \right) $, where $\mc{T}$ is the pre-coalition. 
	Otherwise, some RPPs could have incentives to collude, join together, and collectively interface with the aggregator as a single (and larger) RPP in order to earn a higher total payoff. 
	
	Rigorously, no collusion is defined as follows: $\forall \mc{T} \subseteq \mc{N},$
			\begin{align}  \label{ExAnteNoGamProofPart1}
			\sum_{ i \in \mc{T}} \pi_i \left( c_i^{ne} , \left\{ c_{-i}^{ne} \right\} \right) \geq \tilde{\pi}_{\mc{T}} \left( \tilde{c}_{\mc{T}}^{ne} , \left\{ \tilde{c}_{-\mc{T}}^{ne} \right\} \right), 
			\end{align}
			where %$\left\{ c_i^{\star} , \left\{ c_{-i}^{\star} \right\} \right\}$ is the unique pure NE of the contract game for the case that there is no pre-coalition $\mc{T}$, and the 
			$\left\{ \tilde{c}_{\mc{T}}^{ne} , \left\{ \tilde{c}_{-\mc{T}}^{ne} \right\} \right\}$ is the \emph{new} NE of the new non-cooperative game for the case when the RPPs in $\mc{T}$ join as a single player.  %Note that there are two different contract games: In the first one, each RPP-member of the $\mc{T}$ participates in the contract game as a single RPP, however in the second contract game, RPP-members of $\mc{T}$ form a pre-coalition, and participate in the contract game as \emph{one} RPP. So the unique pure NE of these contract games are not the same.
	
\end{enumerate}

%
%\begin{adjustwidth}{-0.27cm}{}
%\begin{enumerate}
%		
%
%	\item \emph{No-collusion}: This means that if any subset $\mc{T}$ of players want to join together and form a pre-coalition, and then join to the main coalition, the \emph{expected} payoff of this pre-coalition at the unique pure NE of the \emph{new} contract game would be no higher than the sum of those players' expected payoffs at the pure NE of the contract game where they do not form any pre-coalition beofore joining the main coalition, \textit{i.e.} $ \pi_{\mc{T}} \left( c_{\mc{T}}^{\star} , \left\{ c_{-\mc{T}}^{\star} \right\} \right) \leq \sum_{ i \in \mc{T}} \pi_i \left( c_i^{\star} , \left\{ c_{-i}^{\star} \right\} \right) $. We should note that there are two different contract games: one with no pre-coalition, which its unique NE is $\left\{c_i^{\star} , \left\{c_{-i}^{\star}\right\}\right\}$, and the other one is the case that the pre-coalition $\mc{T}$ is formed and participates in the contract game as a single RPP, and its unique pure NE is $\left\{c_{\mc{T}}^{\star} , \left\{c_{-\mc{T}}^{\star}\right\}\right\}$.
%\end{enumerate}
%\end{adjustwidth}

\section{Main Results} \label{mainresults}
We now present the main results of this paper: a proposed payoff allocation mechanism, and how it achieves all the above desired properties 1) - 6). 
First, we propose the following payoff allocation mechanism:  
\begin{align}   \label{NewPAM}
{\mc{P}}_i = 
\begin{cases}
p^f c_i +  		p^{r,b} \left(x_i - c_i\right)       & \quad \text{if }  x_{\mathcal{N}} - c_{\mathcal{N}} < 0 \\
p^f c_i + p^* \left(x_i - c_i\right)       & \quad \text{if }  x_{\mathcal{N}} - c_{\mathcal{N}} = 0\\
p^f c_i +  		p^{r,s} \left(x_i - c_i\right)       & \quad \text{if }  x_{\mathcal{N}} - c_{\mathcal{N}} > 0 
\end{cases},
\end{align}
where $p^{r,s}\le p^*\le p^{r,b}$, and $p^*$ can be chosen arbitrarily within this range. % It can be seen that the propsed PAM boosts those RPPs that reduce the total deviation of the aggregation. 

\begin{RK}[Non-Concavity of the Non-Cooperative Game]
Given the proposed PAM \eqref{NewPAM}, it is worth noting that $\pi_i \left(c_i , \left\{c_{-i}\right\}\right)$ is \emph{not} a concave function in $c_i$. Thus, the non-cooperative game of DA commitments among the RPPs %\textcolor{blue}{, \textit{i.e.} the \textit{contract game},} 
is not a concave game. As such, the behavior of the game (e.g., whether an NE exists) cannot be predicted from existing theories of concave games \cite{rosen1965existence}. Nonetheless, we analyzed the game with new techniques, and the main results are described next. 
\end{RK}

\begin{comment}
Given this PAM, we have that the non-cooperative game among the RPPs possesses a unique pure NE, which moreover can be computed in closed form, as in the following theorem. 
\end{comment}

{
Due to the non-concavity of the non-cooperative game among the RPPs, the existence of a pure NE is not always guaranteed. In the following, we first show the closed form of the unique pure NE if any pure NE exists at all, and then show a necessary and sufficient condition for a pure NE to always exist regardless of the DA and RT prices. 

\begin{theorem} \label{thmNE}
Employing the PAM \eqref{NewPAM}, if the non-cooperative game among the RPPs (cf. Section \ref{noncoop}) possesses any pure NE, it must be unique, and is given by the following DA commitments: $ \forall i = 1, \dots, N$, %submitting $\{c_i\}$
%\vspace{-5pt}
	\begin{align}   \label{closedformNE}
	c_i^{ne}  & =    \mathbb{E} \left[ X_i  \Big| X_{\mathcal{N}} = c_{\mathcal{N}}^{\star} \right],  
	%_{X_i \big| X_{\mathcal{N}} }
	%\hspace{10pt}      \nonumber  
%	&  =  \frac{  \int_{x_i = 0}^{c_{\mathcal{N}}^{\star}} x_i f_{X_{\mathcal{A}}, X_i} \left( x_{\mathcal{A}} = c_{\mathcal{N}}^{\star} \ , \ x_i \right)}{  f_{X_{\mathcal{A}}}  \left(x_{\mathcal{A}} =   c_{\mathcal{N}}^{\star} \right)}  \hspace{5pt}    
	\end{align}
	where $X_\mc{N} = \sum_{i=1}^N X_i$, 
%	\begin{align}  \label{C_N_of_aggregator}
	$c_{\mc{N}}^{\star} \triangleq \argmax_{c_\mc{N}} \expec[\mc{P}_\mc{N}] = \mc{F}_{\mc{N}}^{-1} \left(\frac{p^f - p^{r,s}}{p^{r,b} - p^{r,s}}\right)$, and  
%	\end{align}
%	And $X_{\mc{A}}$ is the random variable indicating the total realization of the aggregator, i.e. $X_{\mc{A}} = \sum_{i \in {\mc{N}}} X_i$. $f_{X_{\mc{A}},X_i} \left(x_{\mc{A}}, x_i\right)$ is the bivariate joint distribution function of $X_{\mc{A}}$ and $X_i$. Also 
$\mc{F}_{\mc{N}} \left(x_{\mc{N}}\right)$ is the cdf %cumulative distribution function 
of $X_\mc{N}$. % We use $\left(c_i^{\star}, \left\{c_{-i}^{\star}\right\}\right)$ to indicate this unique pure NE.
\end{theorem}

\begin{RK}
To compute its equilibrium commitment, an RPP needs only to know % does not need to know the \emph{joint probability distribution of all RPPs} $\left(i.e. \ f_{X_{1}, \dots, X_{N}}  \left(x_1, \dots, x_N\right)  \right)$, which is a strong assumption, to calculate the pure NE defined in \eqref{closedformNE}. Instead, each RPP $i$ needs only to know 
the bivariate probability distribution function (pdf) of its own generation and the total generation of  the aggregation, \textit{i.e.}, $f_{X_i X_{\mc{N}}} \left(x_i , x_{\mc{N}}\right)$. As such, knowledge of the complete joint pdf of all the RPPs is not needed.
\end{RK}

\begin{theorem} \label{conditionNE}
	The non-cooperative game among the RPPs always possesses a pure NE if the following condition holds,
%	if and only if
	\begin{align}  \label{conditionNEeq}
	\forall i \in \mc{N}, \forall \alpha, ~  \frac{d \mathbb{E}_{X_i | X_{\mc{N}}} \left[X_i | X_{\mc{N}} = \alpha\right]}{d \alpha} \leq 1. 
	\end{align}
	Conversely, if the condition \eqref{conditionNEeq} does not hold, then there exists a set of DA and RT prices such that a pure NE does not exist. 
\end{theorem}
%
%\begin{proof}
%The proof of Theorem \ref{conditionNE} is relegated to Appendix \ref{AppdxProofconditionNE}.
%\end{proof}
%
\begin{RK}
We argue that the condition \eqref{conditionNEeq} is a reasonable one, as it holds when no single RPP ``dominates'' the entire aggregation. Consider the case when 
\begin{align} \label{rkeq1}
\exists k \in \mc{N},  \alpha, s.t., \frac{d \mathbb{E}_{X_k | X_{\mc{N}}} \left[X_k | X_{\mc{N}} = \alpha\right]}{d \alpha} > 1.
\end{align} Because $\mathbb{E}_{X_k | X_{\mc{N}}} \left[X_k | X_{\mc{N}}\right] + \sum_{i\ne k} \mathbb{E}_{X_i | X_{\mc{N}}} \left[X_i | X_{\mc{N}}\right]= X_\mc{N}$, we have that 
\begin{align} \label{rkeq2}
\frac{d \mathbb{E}_{X_k | X_{\mc{N}}} \left[X_k | X_{\mc{N}} = \alpha\right]}{d \alpha} + \frac{d \mathbb{E}_{X_{-k} | X_{\mc{N}}} \left[X_{-k} | X_{\mc{N}} = \alpha\right]}{d \alpha} = 1,  
\end{align}
where $X_{-k} \triangleq \sum_{i\ne k} X_i$. From \eqref{rkeq1} and \eqref{rkeq2}, we have 
\begin{align}
\frac{d \mathbb{E}_{X_{-k} | X_{\mc{N}}} \left[X_{-k} | X_{\mc{N}} = \alpha\right]}{d \alpha} < 0. 
\end{align} 
This means that \emph{the aggregation of all the RPPs other than $k$ is negatively correlated with the entire aggregation}. Intuitively speaking, this means that the single RPP $k$ is not only \emph{negatively correlated} with the aggregation of all the other RPPs, but also \emph{dominates} them, and hence dominates the entire aggregation. However, this is an unlikely situation especially when the number of RPPs is relatively large, and no single RPP can dominate the entire aggregation. 
%
%As a result, the following must be true at the particular point around $X_\mc{N} = \alpha$: not only the aggregation of \emph{all the RPPs other than $k$} is negatively correlated with RPP $k$, but their aggregation is ``even ''
%
%
%can be interpreted  has an interpretation. It means that no RPP is large enough to cancel out the effect of the rest of the RPPs. To clarify this, assume that the the total generation of the RPPs changes from $X_{\mc{N}} = \alpha$ to $X_{\mc{N}} = \alpha + \Delta$, where $\Delta > 0$ is a relatively small number. The condition in \eqref{conditionNEeq} means that for all RPP $i$, we have:
%\begin{align*}
%\mathbb{E}_{X_i | X_{\mc{N}}} \left[x_i | X_{\mc{N}} = \alpha + \Delta\right] - \mathbb{E}_{X_i | X_{\mc{N}}} \left[x_i | X_{\mc{N}} = \alpha\right] < \Delta
%\end{align*}
%If for an RPP $k$ this does not hold, it means that 
%\begin{align*}
%\sum_{ i \neq k} \left(\mathbb{E}_{X_i | X_{\mc{N}}} \left[x_i | X_{\mc{N}} = \alpha + \Delta\right] - \mathbb{E}_{X_i | X_{\mc{N}}} \left[x_i | X_{\mc{N}} = \alpha\right]\right) < 0
%\end{align*}
%This means that all the RPPs except the RPP $k$ are positively correlated with each other, but they are negatively correlated with the coalition of them with RPP $k$ (which forms $\mc{N}$). One interpretation of this can be that the size of RPP $k$ is large enough to affect the aggregation of the RPPs.
\end{RK}

For the rest of this paper, we assume that the condition \eqref{conditionNEeq} holds, and thus the unique pure NE of the game is given by \eqref{closedformNE} as in Theorem \ref{thmNE}. %and the contract game has at least one pure NE. In the following Theorem we give a closed form formula for this pure NE, and prove that it is unique.
}
%Theorem \ref{thmNE} 
The closed form of this pure NE immediately implies the following: 
\begin{corollary} \label{coroeff}
Given the PAM \eqref{NewPAM}, efficiency of the aggregation is achieved at the unique pure NE, i.e., 
\begin{align}
c_{\mc{N}}^{ne} = c_{\mc{N}}^{\star},
\end{align}
where $ c_{\mc{N}}^{ne} \triangleq \sum_{i=1}^N c_i^{ne}$. 
\end{corollary}
In other words, Corollary \ref{coroeff} implies that the designed mechanism achieves an ideal ``Price of Anarchy'' of one \cite{AlgorithmGameTheory}. %i.e., the NE does not lose efficiency. %\yue{PoA}

Furthermore, we have that all the remaining desired properties introduced in Section \ref{DesiredPropert} are also achieved: %as in the following theorem. 
\begin{theorem} \label{thmstable}
Employing the PAM \eqref{NewPAM}, individual rationality, stability / in the core, and no collusion are achieved at the unique pure NE specified in \eqref{closedformNE}. 
\end{theorem}
%We note that ex-ante in the core implies individual rationality. 

As a result of Theorems \ref{thmNE}, \ref{conditionNE}, \ref{thmstable} and Corollary \ref{coroeff}, we conclude that the proposed PAM \eqref{NewPAM} induces a \emph{unique NE} among the RPPs, expressed in \emph{closed form} \eqref{closedformNE}, which is both \emph{efficient} and \emph{stable} (i.e., in the core from a coalitional game perspective) for the entire aggregation, and guarantees no collusion. 
We further note that, interestingly, results and techniques similar to our findings have independently been developed for an energy storage sharing problem in the recent works \cite{kalathil2017sharing} and \cite{wu2016sharing}.

%In particular, the proposed PAM achieves an ideal ``Price of Anarchy'' of one (cf. \cite{nisan2007algorithmic}). 
%\begin{RK} 
%\end{RK}

%Lastly, it has been shown in our earlier work \cite{ZhaoISGT2017} that the proposed PAM also achieves all the ex-post properties introduced in Section \ref{ExPostPro}, regardless of what DA commitments $\{c_i\}$ the RPPs submit.  
%
%\yue{transition}
%~
%
%~
%
%~
%
%~

Lastly, while the proposed PAM \eqref{NewPAM} is shown to achieve all the desired properties, an interesting question is how its specific form is discovered. 
For this, we refer the readers to Appendix \ref{sec:market}, %due to space limit, 
in which we show the PAM \eqref{NewPAM} can be derived from a competitive equilibrium of a specially formulated market with transferrable payoff \cite{GameBook}. 

\section{Analysis and Proofs of the Main Results} \label{sec:ana} %The Proposed Mechanism For Aggregation of RPPs} \label{mainresults}
%In this section, 

\subsection{Understanding the Proposed PAM}

\subsubsection{The Excess Payoff from Aggregation}
We first examine the \emph{excess payoff} from aggregating the RPPs given a set of DA commitments $\{c_i\}$, i.e., the difference between a) the realized payoff of the aggregation, and b) the sum of the realized payoffs of the RPPs had they separately participated in the DA-RT market using the \emph{same} DA commitments $\{c_i\}$. 

We define the following notations for the (realization dependent) sets of RPPs with generation surpluses and shortfalls: %respectively. 
\begin{align} %\label{RPPsWithSurplussAndShortfall}
&\mc{S}^+ \triangleq \left\{ i \in \mc{N} ~|~ x_i - c_i \geq 0 \right\}, ~ \mc{S}^- \triangleq \left\{ i \in \mc{N} ~|~ x_i - c_i < 0 \right\}    \nonumber \\
&c_{\mc{S}^+} \! \triangleq \! \sum_{i \in \mc{S}^+} c_i , ~ x_{\mc{S}^+} \!\triangleq\! \sum_{i \in \mc{S}^+} x_i ,~ c_{\mc{S}^-} \!\triangleq\! \sum_{i \in \mc{S}^-} c_i , ~  x_{\mc{S}^-} \!\triangleq\! \sum_{i \in \mc{S}^-} x_i. \nn
\end{align}
For convenience, we define $c_\emptyset = x_\emptyset = 0$. We then have the following lemma on expressing the excess payoff in terms of the above notations, whose proof is relegated to Appendix \ref{AppndxExcssPayoffOfAggrgtr}. 
\begin{lemma} \label{lemexprof}
	The excess payoff from aggregating the RPPs is
	\begin{align}   \label{ExcessProfit}
	\mc{P}_{\mc{N}} &- \sum_{i \in \mc{N}} \mc{P}_i^{sep}    \nonumber \\
	&=\left(p^{r,b} - p^{r,s}\right) \min \left( \left(x_{\mc{S}^+} - c_{\mc{S}^+}\right) , \left(c_{\mc{S}^-} - x_{\mc{S}^-}\right) \right)
	\end{align}
\end{lemma}
%\begin{proof}
%	The proof can be found in Appendix \ref{AppndxExcssPayoffOfAggrgtr}.
%\end{proof}
%
Lemma \ref{lemexprof} implies that the excess payoff from aggregation is \emph{always non-negative}. The excess payoff is zero if the RPPs either all have excesses or all have shortfalls, i.e., no compensation happens among the RPPs.

\subsubsection{Intuition of the Proposed PAM}
To understand the proposed PAM \eqref{NewPAM}, %it is instructive to 
let us consider the following two cases: %, respectively: 
\begin{itemize}
\item \emph{Case 1: The aggregation has a shortfall in total, i.e. $x_{\mathcal{N}} - c_{\mathcal{N}} < 0$}. In this case, %according to \eqref{NewPAM}, 
$\mc{P}_i = \mc{P}_i^{sep}, \forall i\in \mc{S}^-$, i.e., those RPPs with a shortfall earns exactly the same as if they each participates in the market separately. In comparison, $\forall i\in \mc{S}^+,~ \mc{P}_i - \mc{P}_i^{sep} = (p^{r,b} - p^{r,s})(x_i - c_i)$. As a result, only those RPPs in $S^+$ can gain extra earnings compared to if they participate in the markets separately. In other words, \emph{when the aggregation has a shortfall, all the excess payoff \eqref{ExcessProfit} are allocated to those RPPs with a surplus}. 

\item \emph{Case 2: The aggregation has a surplus in total, i.e. $x_{\mathcal{N}} - c_{\mathcal{N}} > 0$}. In this case, $\mc{P}_i = \mc{P}_i^{sep}, \forall i\in \mc{S}^+$, i.e., those RPPs with a surplus earn exactly the same as if they participate in the market separately. In comparison, $\forall i\in \mc{S}^-, ~\mc{P}_i - \mc{P}_i^{sep} = (p^{r,b} - p^{r,s})(c_i - x_i)$. As a result, \emph{when the aggregation has a surplus, all the excess payoff \eqref{ExcessProfit} are allocated to those RPPs with a shortfall}. While this may seem unintuitive at first glance, it can be understood as only rewarding those RPPs who ``reduce the total deviation'', even when the total deviation is a surplus and reducing it means having a shortfall.  

%\item \emph{Case 3: The grand coalition exactly meets its total commitment, i.e. $x_{\mathcal{N}} = c_{\mathcal{N}}$}. In this case, there is a family of payoff allocations that are all in the core: any $p^*$ such that $p^{r,s}\le p^*\le p^{r,b}$ works. Consequently, the PAM that is in the core is not unique in this sense. 
\end{itemize}
\begin{RK}[Marginal Profit]
The proposed PAM \eqref{NewPAM} can also be understood as follows: At RT, given the total (possibly negative) extra generation $x_\mc{N}-c_\mc{N}$ from the aggregation of RPPs, if an additional infinitesimal unit of energy is generated \emph{by RPP $i$}, the resulting additional profit \emph{the aggregation earns} dictates the price of (possibly negative) extra generation $x_i - c_i$ for RPP $i$. 
\end{RK}

\subsection{Proofs of the Main Results} %The Proposed Mechanism Achieves All The Desired Properties
%As it is mentioned in Section  \ref{ExAntePropertDesrdMech}, in order to analyze the ex-ante properties of the PAM, we need to study the  contract game among the RPPs. In the contract game defined in Section \ref{ExAntePropertDesrdMech}, the payoff of the RPPs is defined as the expected value of their realized payoff  \eqref{ContractGamePayoff}, \eqref{NewPAM}:
%

We relegate the proof of Theorems \ref{thmNE} {and \ref{conditionNE}} to Appendices \ref{thm1prf} {and \ref{AppdxProofconditionNE}} due to their mainly algebraic nature. In the following, we present the proof of Theorem \ref{thmstable}. 
	%	\begin{theorem} \label{ExAntePropertsTheoremPart2}
%		Assuming that $p^{r,s} \leq p^f \leq p^{r,b}$, and also assuming that $c_j , \left(\forall j \in \mc{N}\right),$ can be any real number (positive or non-positive), and also assuming that for each RPP $j \ \left(j = 1:N\right)$, given any fixed set of $\left\{c_k\right\}, k \neq j$, the expected payoff of RPP $j$ $\left(i.e. \ \pi_j\right)$ is a continuous function of $c_j$, then the proposed mechanism of this paper satisfies the properties 4) - 6) (cf. Section \ref{DesiredPropert}).
%	\end{theorem}
\begin{proof}[Proof of Theorem \ref{thmstable}]
We will first show that the proposed mechanism \eqref{NewPAM} achieves \emph{individual rationality} and \emph{no-collusion}, which would then be used to prove that \emph{stability / in the core} is further achieved. 
		
%		\begin{lemma}  \label{ExAnteStrongIndivRationLemma}
\begin{enumerate}
\item[a.] \emph{Individual Rationality}: 
%			The proposed PAM \eqref{NewPAM} achieves  individual rationality \eqref{IR} at the unique pure NE of the non-cooperative game among the RPPs.
%		\end{lemma}

%		\begin{proof} 
		Comparing \eqref{PayoffSep} with \eqref{NewPAM}, we immediately have the following inequality: $\forall i\in \mc{N}, \{c_i\}$ and $\{x_i\}, $
			\begin{align}   \label{ExExPostIndRatinlty}
%			& , \nn  \\
			& \hspace{20pt} \mc{P}_i \left( \{c_i\} , \left\{ x_{i} \right\} \right)  \geq  \mc{P}_i^{sep} \left( c_i , x_i \right). 
			\end{align} 
As a heads up, we term \eqref{ExExPostIndRatinlty} \emph{ex-post restricted individual rationality}, which we will describe in detail later in Section \ref{expostpropert} (cf. Property 1 therein). 
			
			%Now we use this inequality and the properties of the Nash equilibrium to prove \emph{individual rationality} property. 
			By taking expectation of \eqref{ExExPostIndRatinlty} over $\{X_i\}$, we have
			\begin{align}   \label{ExAnteIndFromExPost}
			\forall i\in \mc{N} \mbox{ and } \{c_i\} , \ \ \pi_i \left( c_i , \left\{ c_{-i} \right\} \right)  \geq  \pi_i^{sep} \left( c_i \right). 
			\end{align} 
			%Note that the intuition of \eqref{ExAnteIndFromExPost} is that from the ex-post individual rationality we know that for each RPP $i$ and \emph{for any} realization, the payoff from the proposed PAM is greater than the payoff that RPP $i$ would get if he separately participate in the market. 
			We then apply this inequality for the following specific choice of $\{c_i\}: ~ c_i = c_i^{\star, sep} \mbox{ (cf. \eqref{optsep}) and } \{c_{-i}\} = \{c_{-i}^{ne}\}$:   
			\begin{align}  \label{proofExAnteIndivRationLemmaPart2}
			\forall i\in \mc{N}, ~~\pi_i \left( c_i^{\star , sep} , \left\{ c_{-i}^{ne} \right\} \right)  \geq   \pi_i^{sep} \left( c_i^{\star , sep} \right). 
			\end{align}
			
			On the other hand, from the \emph{best responding} property in the definition of NE \eqref{defNE}, we have 
			\begin{align}  \label{proofExAnteIndivRationLemmaPart1} 
			\pi_i \left( c_i^{ne} , \left\{ c_{-i}^{ne} \right\} \right)  \geq  \pi_i \left( c_i^{\star , sep} , \left\{ c_{-i}^{ne} \right\} \right)
			\end{align}
			Combining \eqref{proofExAnteIndivRationLemmaPart1} and \eqref{proofExAnteIndivRationLemmaPart2}, we have 
			\begin{align}  \label{proofExAnteIndivRationLemmaPart3}
			\pi_i \left( c_i^{ne} , \left\{ c_{-i}^{ne} \right\} \right)  \geq  \pi_i^{\star , sep} \left( c_i^{\star , sep} \right),
			\end{align}
			which completes the proof of individual rationality. 
%		\end{proof}
		%
		%
		\item[b.] \emph{No Collusion}:
%		\begin{lemma}  \label{ExAnteNoGamLemma}
%			The proposed PAM \eqref{NewPAM} achieves no-collusion at the unique pure NE of the non-cooperative game among the RPPs.
%		\end{lemma}
		%
%		\begin{proof}

			%
			%We use the properties of the unique pure NE defined in \eqref{closedformNE} and the linearity of the proposed PAM \eqref{NewPAM} to prove this property.
			
			To prove no collusion \eqref{ExAnteNoGamProofPart1}, we begin with showing that 
			\begin{align} 
			\tilde{c}_{\mc{T}}^{ne}  =  \sum_{ i \in \mc{T}} c_i^{ne}, ~\forall \mc{T\subseteq \mc{N}}. 
			\end{align} 
			Without loss of generality (WLOG), consider RPPs $1, 2, \cdots, T$ form as a joint player $\mc{T}$. 
			Based on Theorem \ref{thmNE}, from the closed form expression of the unique pure NE \eqref{closedformNE}, we have
			\begin{align}
			\tilde{c}_{\mc{T}}^{ne} &= \mathbb{E} %_{X_{\mc{T}} \big| X_{\mc{A}} } 
			\left[ X_{\mc{T}}  \Big| X_{\mc{N}} = c_{\mc{N}}^{\star} \right]  =  \mathbb{E}%_{ \sum_{ j \in \mc{T}} X_j \big| X_{\mc{A}} } 
			\left[ \sum_{ j \in \mc{T}} X_j \Big| X_{\mc{N}} = c_{\mc{N}}^{\star} \right] \nonumber \\ 
			&= \sum_{ j \in \mc{T}} \mathbb{E}%_{ \sum_{ j \in \mc{T}} X_j \big| X_{\mc{A}} } 
			\left[ X_j \Big| X_{\mc{N}} = c_{\mc{N}}^{\star} \right]  %\sum_{ j \in \mc{T}}
			= \sum_{ i \in \mc{T}} c_i^{ne}. 
			\end{align}
			Similarly, in the \emph{new} game with RPPs in $\mc{T}$ %$1,\ldots,T$ 
			joining as a single player, the \emph{remaining} RPPs' commitments at the new NE stay \emph{the same} as at the original game's NE:
			\begin{align}
			\forall i\notin \mc{T}, ~ \tilde{c}_i^{ne} = \mathbb{E} \left[ X_i  \Big| X_{\mathcal{N}} = c_{\mathcal{N}}^{\star} \right] = {c}_i^{ne}. 
			\end{align}
			Now, from the special piece-wise linear structure of the proposed PAM \eqref{NewPAM}, it is straightforward to verify that
			\begin{align}
			\tilde{\mc{P}}_{\mc{T}} &\left( \tilde{c}_{\mc{T}}^{ne}, \left\{\tilde{c}_{-\mc{T}}^{ne}\right\} , \left\{x_j \right\} \right)  %\nn\\
			= \sum_{ i \in \mc{T}} \mc{P}_i \left( \{c_i^{ne}\}, \{x_j\} \right). \nn
			\end{align} 
			%since $ c_{\mc{T}}^{\star}  =  \sum_{ i \in \mc{T}} c_i^{\star} $ and $ x_{\mc{T}}  =  \sum_{ i \in \mc{T}} x_i $, one can see that for the at the unique pure NE, for any set of realizations $\left\{x_j\right\}, \ j \in \mc{N}$, we have: $\sum_{ i \in \mc{T}} \mc{P}_i \left( c_i^{\star}, \left\{c_{-i}^{\star}\right\} , x_i, \left\{x_{-i}\right\} \right)  = \mc{P}_{\mc{T}} \left( c_{\mc{T}}^{\star}, \left\{c_{-\mc{T}}^{\star}\right\} , x_{\mc{T}}, \left\{x_{-\mc{T}} \right\} \right)$. This comes from the linearity of the payment allocation model (PAM) of the proposed mechanism defined Section \ref{ProposedMechanism}.
			In other words, the total payoff allocated to the RPPs in $\mc{T}$ %$1,\ldots,T$ 
			remains \emph{the same} before and after they form as a joint player. %We note that 
			As this holds \emph{in all circumstances}, %which immediately implies that 
			it also holds \emph{in expectation}:  
			\begin{align}   \label{ExAnteNoGamProof}
			\hspace{0pt} \tilde{\pi}_{\mc{T}} \left(\tilde{c}_{\mc{T}}^{ne} , \left\{ \tilde{c}_{-\mc{T}}^{ne} \right\}\right) = \sum_{i \in \mc{T}} \pi_i \left(c_i^{ne}, \left\{c_{-i}^{ne}\right\}\right). 
			\end{align}
			As a result, the proposed PAM \eqref{NewPAM} achieves the no collusion property \eqref{ExAnteNoGamProofPart1}: In particular, the inequality is always achieved by equality \eqref{ExAnteNoGamProof}. We term equation \eqref{ExAnteNoGamProof} the \emph{No Collusion Equation}. 
			\item[c.] \emph{Stability / In the Core:} ~
			We now show that individual rationality and no collusion collectively implies stability of the proposed PAM (cf. Property 5 in Section \ref{DesiredPropert}). 

			For any subset of RPPs $\mc{T}\subseteq \mc{N}$, consider the hypothetical case of them \emph{joining as a single player} to aggregate with the remaining RPPs $\mc{T}\backslash \mc{N}$ under the proposed PAM \eqref{NewPAM}. \emph{Applying individual rationality (cf. \eqref{IR}) specifically to this joint player}, we have 
			\begin{align}  \label{ExAnteInCorePart2}
			\tilde{\pi}_{\mc{T}} \left( \tilde{c}_{\mc{T}}^{ne} , \tilde{c}_{-\mc{T}}^{ne} \right)  \geq  \pi_{\mc{T}}^{\star , sep} \left( c_{\mc{T}}^{\star , sep} \right). 
			\end{align}
			From \eqref{ExAnteInCorePart2} and the No Collusion Equation \eqref{ExAnteNoGamProof}, we have %finally 
			\begin{align}  \label{ExAnteInCorePart3}
			\sum_{i \in \mc{T}} \pi_i \left(c_i^{ne}, \left\{c_{-i}^{ne}\right\}\right) \geq  \pi_{\mc{T}}^{\star , sep} \left( c_{\mc{T}}^{\star , sep} \right), 
			\end{align}
			completing the proof of Property 5) - Stability / In the Core - of the proposed PAM \eqref{NewPAM}. 
		
\end{enumerate}
%This concludes the proof of Theorem \ref{thmstable}. 
\vspace{-13pt}
			\end{proof}
	
\section{Ex-post Properties of The Proposed Mechanism}  \label{ExPostPropThePropMechan}
In this section, we present another set of desirable properties achieved by the proposed PAM \eqref{NewPAM}: These properties are termed ``ex-post'' properties because they are achieved for \emph{all possible realizations} of the random power generation $\{X_i\}$. These properties, however, are distinguished from those discussed in Section \ref{DesiredPropert} by a key caveat --- an assumption on the RPPs' DA commitments, as will be described next. 

\subsection{A Specialized Coalitional Game}

{Given any set of DA commitments $\{c_i\}$ and realizations of generation $\{x_i\}$}, similar to the realized payoffs \eqref{PayoffSep} and \eqref{PayoffAgg}, we define a function $v(\cdot)$ for the value of a coalition of any subset of RPPs as follows: $\forall \mathcal{T}\subseteq\mc{N}$,
\begin{align}  \label{ValOfCoal}
v \left(\mathcal{T}\right) &= p^f \hat{c}_\mathcal{T} - p^{r,b} \left(\hat{c}_\mathcal{T} - x_\mathcal{T}\right)_+ + p^{r,s} \left(x_\mathcal{T} - \hat{c}_\mathcal{T}\right)_+ 
\end{align}
where $\hat{c}_{\mathcal{T}} \triangleq \sum_{i \in \mathcal{T}}{c_i}$ and $x_{\mathcal{T}} \triangleq \sum_{i \in \mathcal{T}}{x_i}$. 
With the above value function $v(\cdot)$, designing the PAM $\{\mc{P}_i(\{c_i\}, \{x_i\}), i \in\mc{N}\}$ %given any fixed $\{c_i\}$ and $\{x_i\}$, 
can be studied in a well-defined coalitional game \cite{GameBook}. In particular, in a coalitional game, a PAM $\{\mc{P}_i\}$ is said to be \emph{stable/in the core} if and only if it satisfies the following set of inequalities: 
	\begin{align}   \label{DefCoreGame}
	\forall \mc{T} \subseteq \mathcal{N}, ~  \sum_{i \in \mc{T}} \mc{P}_i \geq v\left(\mc{T}\right). 
	\end{align}
%Now, note that a PAM $\{\mc{P}_i(\{c_i\}, \{x_i\})\}$ is defined as a set of functions of \emph{any possible DA commitments $\{c_i\}$ and realizations of generation $\{x_i\}$}. 
In other words, the total payoffs allocated to any subset of the players should be no less than the value of this subset. In particular, here the value of a subset of RPPs $\mc{T}$ \eqref{ValOfCoal} has a specific meaning --- the realized payoff of the subset of RPPs $\mc{T}$ had they left the aggregation and collectively participate in the DA-RT markets \emph{with a specific total commitment --- $\hat{c}_{\mathcal{T}} = \sum_{i \in \mathcal{T}}{c_i}$}. As such, if a PAM $\{\mc{P}_i\}$ is in the core (cf. \eqref{DefCoreGame}), any subset of RPPs $\mc{T}$ in the aggregation earn at least as much as they would otherwise earn outside the aggregation \emph{provided that they stick to the same total DA commitments as they do inside the aggregation.}

\begin{RK}[Restrictive Assumption on DA Commitments] \label{RKrestr}
As with the previously discussed stability property in Section \ref{DesiredPropert}, one would ideally like stability / in the core be satisfied \emph{without any restriction} on how a subset of RPPs determine their DA commitments. The requirement of $\hat{c}_{\mathcal{T}} = \sum_{i \in \mathcal{T}}{c_i},~\forall \mc{T}\subseteq\mc{N}$ in this section is hence a \emph{restrictive} one, as it does not allow a subset of RPPs leaving the aggregation to re-adjust their DA commitments. This assumption nonetheless leads to a set of ``ex-post'' properties of the proposed PAM \eqref{NewPAM} in the following. 
\end{RK}

\subsection{Ex-post Restricted Stability and No Collusion} \label{expostpropert}
We now describe the properties achieved by the proposed PAM \eqref{NewPAM} in an ``ex-post'' sense, meaning that they hold for all possible $\{c_i\}$ and renewable generation realizations $\{x_i\}$. 
\begin{enumerate}
	\item \emph{Ex-post restricted individual rationality}: $\mathcal{P}_i\left(\{c_i\}, \{x_i\}\right)  \geq \mathcal{P}_i^{sep}(c_i,x_i), ~\forall i \in\mc{N},$ (cf. \eqref{PayoffSep} and \eqref{NewPAM}, and mentioned earlier as \eqref{ExExPostIndRatinlty}). In other words, the payoff of RPP $i$ is at least as high as the payoff it could have gotten had it separately participated in the DA-RT market \emph{with the same DA commitment $c_i$ as originally submitted to the aggregator}. 
	\item \emph{Ex-post restrictedly stable/in the core}: Being restrictedly stable/in the core in an ``ex-post'' sense is defined by \eqref{DefCoreGame} and \eqref{ValOfCoal}. In other words, 
if any subset of the RPPs $\mc{T}$ leave the aggregation, separately form their own aggregation, and then participate in the market \emph{with the same sum of DA commitments $\hat{c}_{\mathcal{T}} = \sum_{i\in\mc{T}}c_i$ as originally submitted to the aggregator}, they would get a realized payoff no higher than the sum of their realized payoffs originally from the PAM. 
	\item \emph{Ex-post restricted no collusion}: Suppose any subset of RPPs $\mc{T}$ join together as a single player before participating in the aggregation with the remaining RPPs, and submit the same sum of DA commitments $\hat{c}_{\mathcal{T}} = \sum_{i\in\mc{T}}c_i$ to the aggregator. Their total realized payoff would be \emph{no higher} than the sum of their original realized payoffs from the PAM. Rigorously, $\forall \mc{T}$, 
\begin{align}
 \sum_{i \in \mc{T}} \mc{P}\left(\{c_i\}, \{x_i\}\right) \ge \tilde{\mc{P}}_\mc{T} \left(\hat{c}_\mc{T}, \{c_{-\mc{T}}\}, \{x_i\}\right), 
\end{align}
where $\tilde{\mc{P}}_\mc{T} \left(\hat{c}_\mc{T}, \{c_{-\mc{T}}\}, \{x_i\}\right)$ is the \emph{new} payoff of a subset of RPPs $\mc{T}$ if they a) join as a single player, and then b) aggregate with the remaining RPPs $\mc{N} \backslash \mc{T}$ under the proposed PAM \eqref{NewPAM}, employing the original total commitment of $\hat{c}_{\mathcal{T}} = \sum_{i\in\mc{T}}c_i$.

\end{enumerate}
We note that the reason for the above properties to be termed as ``restricted'' ones is the assumption of $\hat{c}_{\mathcal{T}} = \sum_{i\in\mc{T}}c_i,~\forall \mc{T}\subseteq\mc{N}$ (cf. Remark \ref{RKrestr}). %all RPPs must submit their original DA commitments (cf. Remark \ref{RKrestr}) had they left the aggregation (cf. Properties 1 and 2) or form as. 
We now have the following theorem. 
\begin{theorem} \label{expostthm}
Employing the PAM \eqref{NewPAM}, ex-post restricted individual rationality, stability / in the core, and no collusion are achieved for all possible $\{c_i\}$ and $\{x_i\}$. 
\end{theorem}

It is straightforward to verify from \eqref{PayoffSep} and \eqref{NewPAM} that ex-post restricted individual rationality and no collusion both hold. To prove ex-post restricted stability/in the core, we again refer the readers to Appendix \ref{sec:market}, in which stability/in the core is implied by the competitive equilibrium of a specially formulated market with transferrable payoff. 
We further note that, interestingly, results similar to the proposed PAM \eqref{NewPAM} in achieving ex-post properties have independently been developed in a recent work \cite{chakraborty2017cost}, using a cost causation based analysis. 
%we provide an equivalent formulation of the coalitional game defined in Section \ref{} as a \emph{Market with Transferrable Payoff} \cite{}. We then show that the \emph{competitive equilibrium of the market is exactly given by the proposed PAM \eqref{NewPAM}}. As a result, the competitive equilibrium immediately implies that stability / in the core is satisfied, and hence proves Theorem \ref{}. The details are relegated to Appendix \ref{}. 

%the intuition of \eqref{DefCoreGame} is exactly as described above in property 5). If \eqref{DefCoreGame} is satisfied for all possible realizations of the random renewable generation, the PAM $\{\mc{P}_i\}$ is said to be in the core in the ex-post sense. 

%			
%			
%
%~
%
%~
%
%~
%
%~
%
%~

%
%
%\begin{RK}   \label{ExAnteHoldAtNE_ExPostHoldAlways}
%	It is very important to note that the (ex-ante) properties 4) - 6) in Section \ref{DesiredPropert} are hold \emph{at equilibrium}. However the ex-post properties 1) - 6) hold for \emph{all} possible sets of power contracts $\left\{c_i\right\}$ and all possible sets of realizations $\left\{x_i\right\}$. Although the ex-post properties 4) - 6) are restricted versions of their ex-ante counter parts, but they hold for all possible scenarios.
%\end{RK}
%

%

%
\section{Simulation}  \label{Simul}
\subsection{Data Description and Simulation Setup} \label{simmodel}

We perform simulations using the NREL dataset \cite{NREL2010} for ten wind power producers (WPPs) located in the PJM interconnection. For each WPP, both the hourly DA forecasts and the actual realized generation are available from the data set. 
The generation of the WPPs for each hour $t$ are modeled as 
\begin{align}
W_i \left(t\right) = \widehat{W}_i \left(t\right) + \epsilon_i \left(t\right), ~\forall i, \nn  
\end{align}
where $ \widehat{W}_i $ is the (point) forecast generation of WPP $i$, and $ \epsilon_i$ is the forecast error. For simplicity, we consider the WPPs modeling their forecast errors using a zero mean jointly Gaussian distribution, $N(0, \Sigma)$. 
We fit the covariance matrix $\Sigma$ using the real data of these ten WPPs in Jan. 2004. 
The simulations are then performed based on the real data of these ten WPPs in Feb. 2004. We note that the WPPs' (Gaussian) probabilistic beliefs are only their crude statistical models of their generation, and all the data used in the actual simulations are real data (as opposed to Gaussianly distributed). This, however, is already sufficient to provide instructive numerical results as will be shown in the remainder of the section. 

To simulate the WPPs' interactions with the DA-RT markets, we employ the hourly DA and RT locational marginal prices (LMPs) in Feb. 2004 from where the ten WPPs are located (all in the PJM interconnection). %As these ten WPPs are located close to each other, the same LMPs. 
In particular, $p^f$ in \eqref{PayoffSep} is obtained from the hourly day ahead market price $p^{DA} (t)$. To obtain $p^{r,b}$ and $p^{r,s}$, the same approach as in \cite{ZQRGP15} is employed: we let $p^{r,b} = \max \left( 1.2 p^{DA} (t) , 2 p^{RT} (t) \right)$ and $p^{r,s} = \min \left( p^{DA} (t) / 1.2 ,  p^{RT} (t) / 2 \right)$, where $p^{RT} (t)$ is the hourly real time market price. 

%%Note that this way of defining the prices guarantees the condition of $p^{r,s} \leq p^f \leq p^{r,b}$.
%The intuition behind this price model is the following: In DA market, WPPs are motivated to act conservatively because even if they could predict real-time prices (of the next day) accurately, they are incentivised to consider higher penalty rates for shortfal and lower payments for excess generation.

%Note that for all of these cases we need to know the statistical distribution of the WPPs. We assume that each WPP uses a normal distribution function to model its future generation: For each hour of the next day, a) a WPP's forecast for that future hour, available from the NREL dataset \cite{NREL2010}, is taken as the the mean of such a distribution, and b) the variance of the distribution is approximated using the data of that WPP during Jan. 2004. We further assume that the bivariate probability distribution function of each WPP and the aggregator is a bivariate normal distribution function. This helps us to use the simplified formula in \emph{Corollary} \ref{coro3}.

We evaluate four different cases of WPPs participating in the DA-RT market, with and without aggregation:  
%\begin{adjustwidth}{-0.25cm}{}
\begin{itemize}
	\item Case 1: An aggregator employs an efficient and stable/in the core PAM previously derived in \cite{ZQRGP15} to aggregate the WPPs. This PAM assumes the knowledge of the joint probability distribution of the WPPs' random generation. 
	\item Case 2:  An aggregator employs the proposed mechanism (cf. Section \ref{indmec} and \eqref{NewPAM}) to aggregate the WPPs. At the unique NE, %\emph{Nash equilibrium}, 
	each WPP $i$ submits $c_i^{ne}$ (cf. \eqref{closedformNE}). 
	\item Case 3: An aggregator employs the proposed mechanism to aggregate the WPPs. Each WPP $i$ submits the DA commitment that would be optimal had it separately participated in the markets, i.e., $c_i^{\star , sep} = \argmax_{c_i} \pi_i^{sep} \left(c_i\right)$. %, which can be solved as a news-vendor problem \cite{Bitar2012}. 
	\item Case 4: Without an aggregator, each WPP $i$ separately participates in the DA-RT market, and makes the optimal DA commitment $c_i^{\star , sep} = \argmax_{c_i} \pi_i^{sep} \left(c_i\right)$. 
\end{itemize}
%\end{adjustwidth}

In particular, for Case 1, the simulated payoff allocation mechanism based on the PAM in [9] is given by 

\begin{align}  \label{ex_post_Zhao_2015}
\mc{P}_i^{ \mbox{ case }1} =  p_i^* x_i + \frac{\mc{P}_{\mc{N}} - \sum_{j=1}^N p_j^* x_j}{N}, 
%\frac{\pi_i^{\mbox{ case }1}}{\mathbb{E}X_i}
\end{align}
where $p_i^*$ is the ``competitive price'' given by eq. (15) in \cite{ZQRGP15}. 

%\begin{align}  \label{ex_post_Zhao_2015}
%\mc{P}_i^{ \mbox{ case }1} = \frac{\pi_i^{\mbox{ case }1}}{\sum_{j \in \mc{N}} \pi_j^{\mbox{ case }1}} \cdot \mc{P}_{\mc{N}}, 
%\end{align}
%where $\pi_i^{\mbox{ case }1}$ is given by eq. (17) in [9]. 

\subsection{Simulation Results}

%%%%%%%
\begin{table}[t]
	\centering
	\caption{Total payoff of all the WPPs}
%		\vspace{-10pt}
		\small
	\begin{tabular}[l]{| p{2.0cm} | p{1.8cm} | p{1.5cm} | p{1.4cm} |}
		\hline
		 & Cases 1 and 2 & Case 3 & Case 4 \\
		\hline
		 Total Payoff (\$) & $10,428,257$  &  $10,352,581$ & $9,148,024$  \\
		\hline
	\end{tabular}
\label{table:Table1}
\normalsize
\end{table}

The total payoffs of all the WPPs for the four cases are summarized in Table. \ref{table:Table1}. As expected from Corollary \ref{coroeff}, the total payoffs for Cases 1 and 2 are the same since both cases achieve \emph{efficiency, i.e., maximum expected profit} for the aggregation. In comparison, since Case 3 does not achieve efficiency for the aggregation, a lower total payoff is achieved than that in Cases 1 and 2. Lastly, Cases 1, 2, and 3 all achieve significantly higher total payoffs than Case 4, demonstrating the benefit of aggregating the WPPs. 

Breaking down the total payoffs across WPPs and hours, a) the daily average payoffs of the WPPs are shown in Figure \ref{WPPpayoffs}, and b) the hourly average payoffs of the aggregation are shown in Figure \ref{Hourpayoffs}. It is observed that the payoffs in Cases 1 and 2 are consistently higher than that in Case 3, which are further always higher than that in Case 4. %Accordingly, with only a crude statistical model of the forecast errors as in the preceding subsection, numerical results consistent with the theoretical predictions (cf. Section \ref{mainresults}) have been observed. 
For all the WPPs, individual rationality is confirmed. 

\begin{figure}[t!]
	\centering
	\includegraphics[scale=0.40]{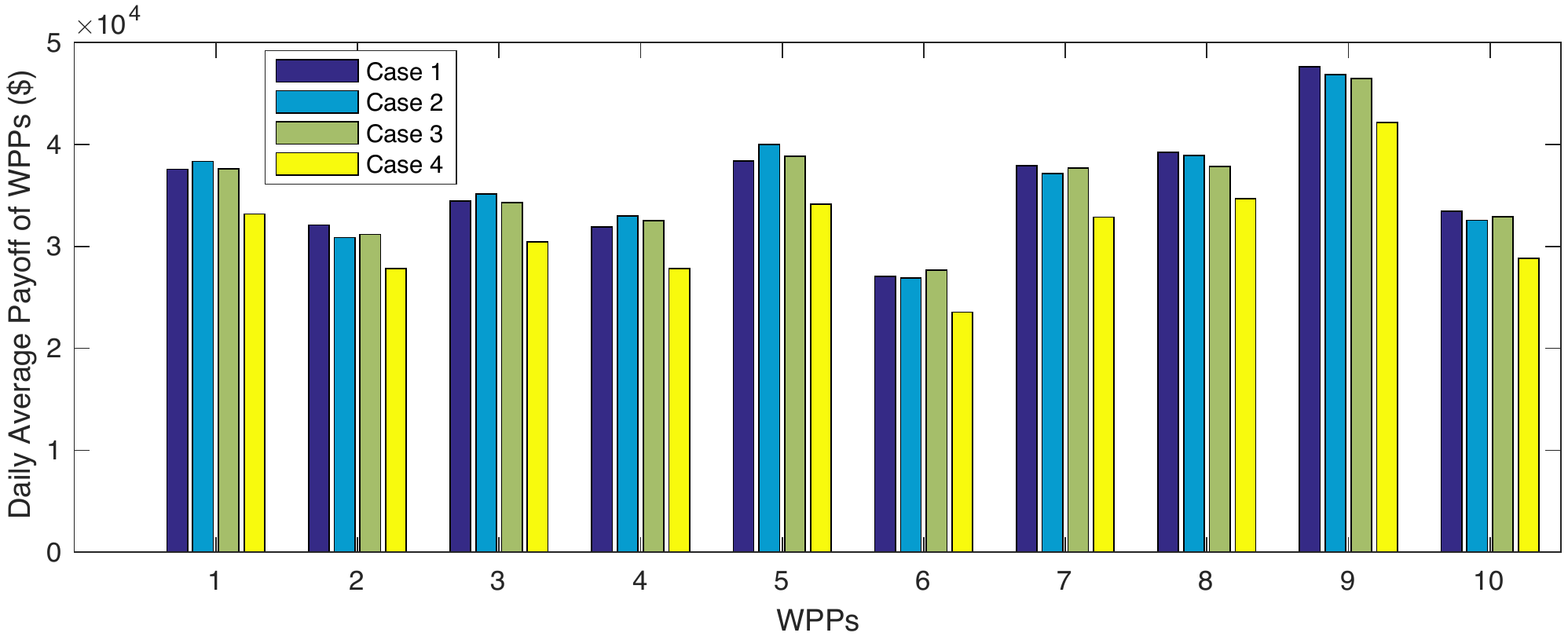}
	\caption{Comparison of the daily average payoffs of the WPPs.}
	\label{WPPpayoffs}
\end{figure}

\begin{figure}[t!]
	\centering
	\includegraphics[scale=0.48]{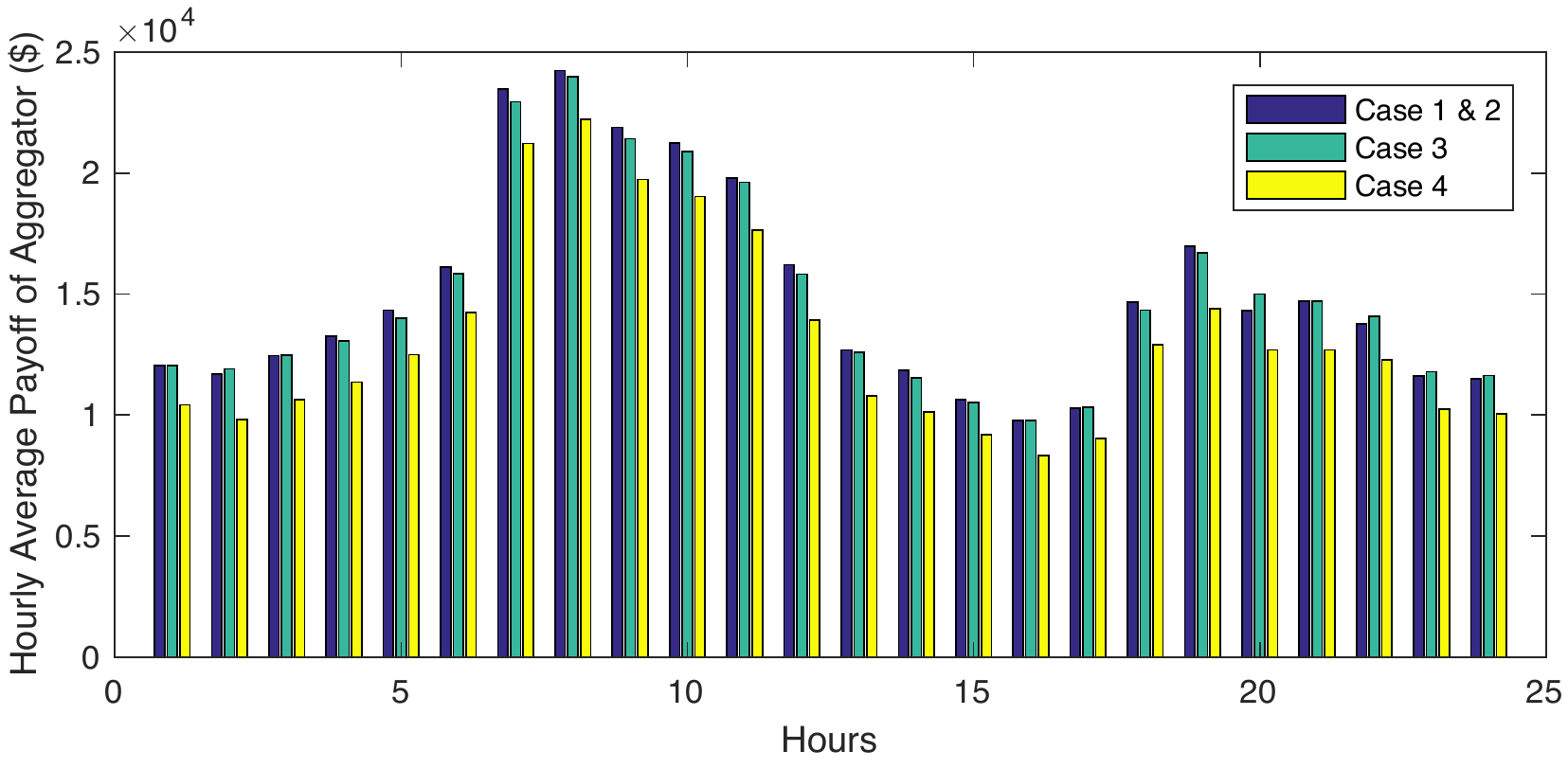}
	\caption{Comparison of the hourly average payoffs of the aggregation.}
	\label{Hourpayoffs}
\end{figure}

Furthermore, we plot the real-time payoff traces of one of the ten WPPs (\#$05711$) in Figure \ref{traces}. 
\begin{itemize}
\item In Figure \ref{exantetrace}, we compare Case 2 against Case 4: it is observed that, while for most of the times the WPP earns a higher payoff in Case 2, there are a small number of hours (e.g., hour \#144 and \#189) in which the WPP earns a higher payoff in Case 4. This is not unexpected for two reasons: a) Case 2 achieves individual rationality \emph{in expectation} (cf. Section \ref{DesiredPropert}), and therefore does not preclude some realized scenarios in which Case 4 turns out better, and b) The WPPs and the aggregator only employ a crude Gaussian model for the forecast errors in our simulations (cf. Section \ref{simmodel}), and thus the WPPs' payoffs determined accordingly can potentially deviate from the ideal ones had the ``ground truth'' probability distributions are employed. 
\item In Figure \ref{exposttrace}, we compare Case 3 against Case 4: it is observed that for \emph{all times} the WPP earns a higher payoff in Case 3. This is consistent with the \emph{ex-post} restricted individual rationality (cf. Section \ref{expostpropert}), in particular because each WPP's in Case 3 makes the DA commitments \emph{in the same way} as in Case 4 ($c_i^{\star , sep}$). 
\end{itemize}
Similar observations have been made for the realized real-time payoff traces for the \emph{entire aggregation}: a) In Case 2, \emph{for most but not all times}, the aggregator's total payoff is higher than that in Case 4, and b) In Case 3, \emph{for all times} the aggregator's total payoff is higher than that in Case 4. In light of these, it is worth re-emphasizing that \emph{efficiency, i.e., maximum expected profit} for the aggregation is in fact achieved in Case 2, but not in Case 3 (cf. Table \ref{table:Table1}), even though for some realized scenarios the payoffs in Case 2 are worse than Case 3. 

\begin{figure}[t!]
	\centering
	\subfigure[]{
	\includegraphics[scale=0.44]{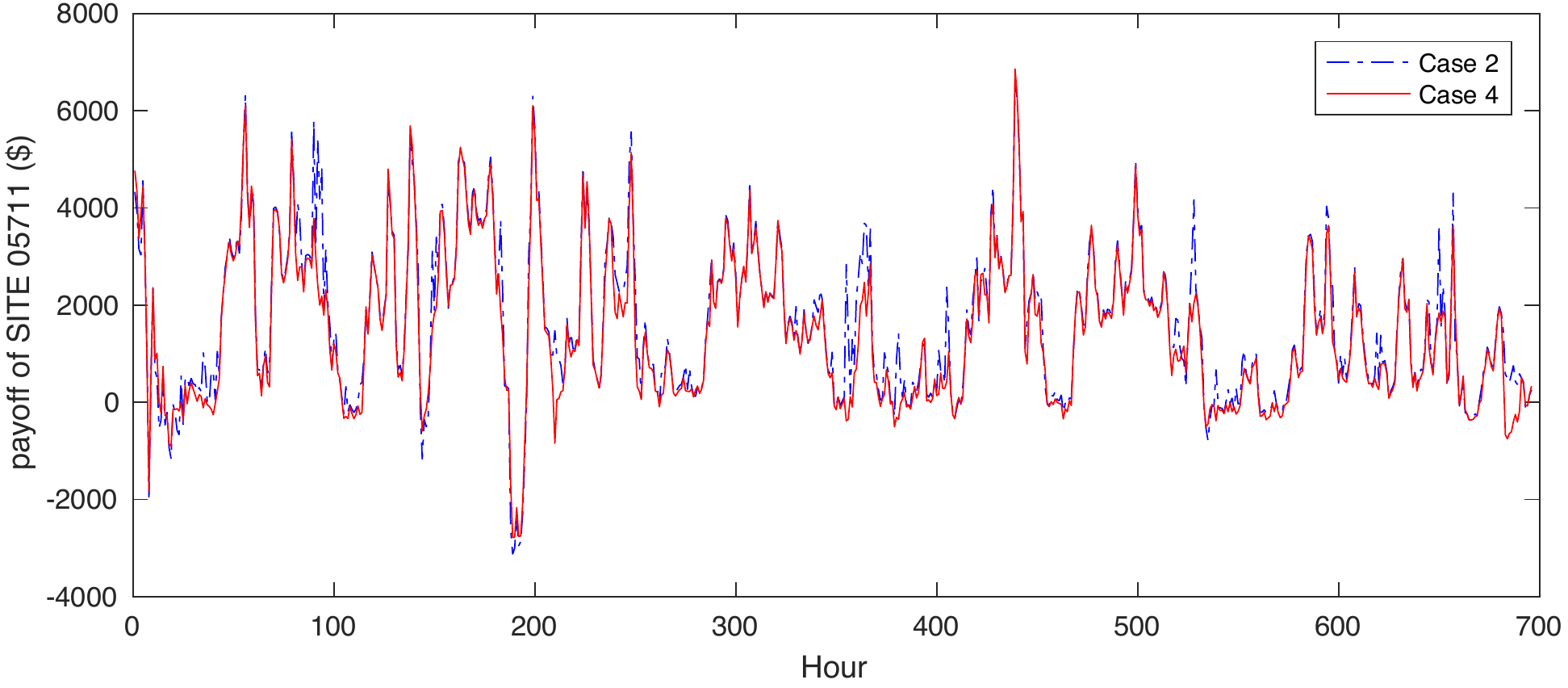} 
	\label{exantetrace}
	}
	\subfigure[]{
	\includegraphics[scale=0.44]{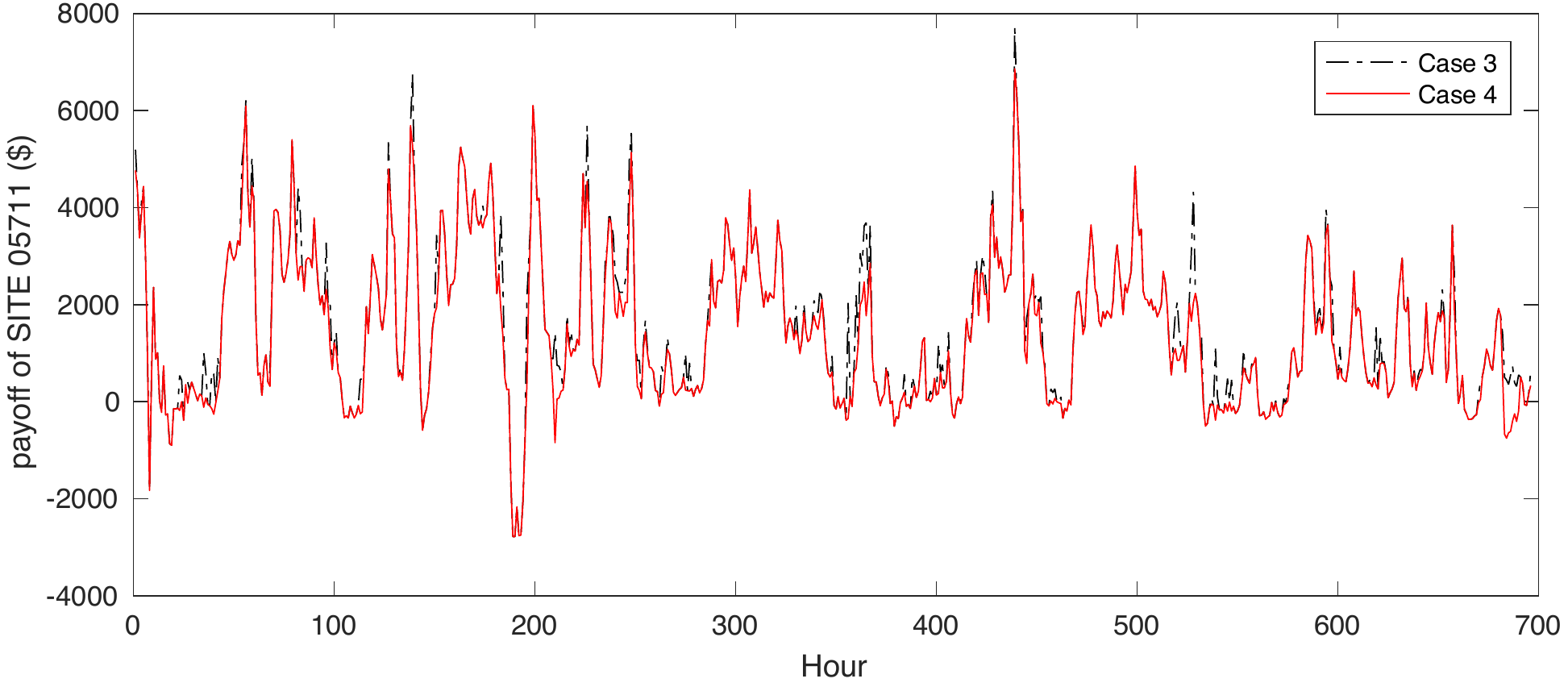}
	\label{exposttrace}
	}
	\caption{The payoffs traces over time of a WPP.}
	 \label{traces}
\end{figure}

\section{Conclusion}    \label{Concl}
An indirect mechanism design framework is employed for aggregating renewable power producers in a two settlement power market. We have designed a payoff allocation mechanism by solving the competitive equilibrium of a specially formulated market with transferrable payoff. 
%3 steps
We have proved that the outcome of the designed mechanism is predicted by a unique Nash equilibrium among the RPPs participating in the aggregation, characterized in closed from. Moreover, at this NE induced by the mechanism, the entire aggregation achieves efficiency, i.e., the maximum expected profit as if all the RPPs fully cooperate. This implies an ideal ``Price of Anarchy'' of one. We have then proved that the NE is ``in the core'', and is hence stable from a coalitional game theoretic perspective. Furthermore, we have proved that the NE guarantees no collusion among the RPPs. In addition, a set of ``ex-post'' properties are also achieved by the designed mechanism. We have simulated the designed mechanism with data from 10 wind power producers in the PJM interconnection. Numerical results consistent with theoretical predictions are observed. 

\section*{Acknowledgment}
The authors would like to thank the anonymous reviewers for their helpful comments and suggestions. %which have been very helpful. 

\ifCLASSOPTIONcaptionsoff
\newpage
\fi

\bibliographystyle{IEEEtran}
{\bibliography{TPS}}

%\newpage
%~

%\newpage

\appendices

\section{Proof of Lemma \ref{lemexprof} }  \label{AppndxExcssPayoffOfAggrgtr} %The Excess Payoff of The Aggregator
First, we have that
\begin{align}   \label{ExcssProfitSumX}
&\mc{P}_i^{sep} = \begin{cases}
p^f c_i + p^{r,s} \left(x_i - c_i\right),       & \quad \text{if }  i \in \mc{S}^+ \\
p^f c_i - p^{r,b} \left(c_i - x_i\right).        & \quad \text{if }  i \in \mc{S}^- \\
\end{cases}   \nonumber  
\end{align}
As a result, 
\begin{align}
&\sum_{i \in \mc{N}} \mc{P}_i^{sep} = \sum_{i \in \mc{S}^+} \mc{P}_i^{sep} + \sum_{i \in \mc{S}^-} \mc{P}_i^{sep}  \nonumber \\
&=p^f \left(c_{\mc{S}^+} + c_{\mc{S}^-}\right) - p^{r,b} \left( c_{\mc{S}^-} - x_{\mc{S}^-} \right) + p^{r,s} \left( x_{\mc{S}^+} - c_{\mc{S}^+} \right) 
\end{align}
We now consider the case of
$x_{\mc{S}^+} - c_{\mc{S}^+}  \geq  c_{\mc{S}^-} - x_{\mc{S}^-}$, i.e., there is an excess power in total in the aggregation. In this case, 
\begin{align}   \label{ExcssProfitAgg1X}
&\mc{P}_{\mc{N}} = p^f \left(c_{\mc{S}^+} + c_{\mc{S}^-}\right) + p^{r,s} \left(x_{\mc{S}^+} + x_{\mc{S}^-} - c_{\mc{S}^+} - c_{\mc{S}^-}\right)
\end{align}
From (\ref{ExcssProfitSumX}) and (\ref{ExcssProfitAgg1X}), we have:
\begin{align}
&\mc{P}_{\mc{N}} - \sum_{i \in \mc{N}} \mc{P}_i^{sep} = (p^{r,b} - p^{r,s}) \left(c_{\mc{S}^-} - x_{\mc{S}^-}\right) \nonumber \\
&= \left(p^{r,b} - p^{r,s}\right) \min \left( \left(x_{\mc{S}^+} - c_{\mc{S}^+}\right) , \left(c_{\mc{S}^-} - x_{\mc{S}^-}\right) \right)
\end{align}

The case when $x_{\mc{S}^+} - c_{\mc{S}^+}  <  c_{\mc{S}^-} - x_{\mc{S}^-}$ can be proved similarly.

%
%\section{Comment on Negative Values of $c_i^{\star}$}   \label{AppenNegatC}
%In Section \ref{ProblmInRealWrld} we allowed the RPPs to submit negtive values for their power contracts. Negative value for power contract somehow means that the RPP acts like a load in the day-ahead market. The main reason of this assumption is to ensure that the contract game defined in Section \ref{ContractGame} possesses an unique pure Nash equilibrium (NE). Now we show that at this pure NE defined in \eqref{closedformNE}, the power contracts of the RPPs (\textit{i.e.} $\left\{c_i\right\}$) are non-negative. If we take a close look at \eqref{closedformNE}, we notice that the $c_i^{\star}$ is always non-negative. Note that the reason that the lower bound of the integral (in the denominator) is $0$ (not $-\infty$) is the fact that we do not have negative generations. So at this unique pure NE of the contract game, the RPPs' power contracts are non-negative.
%
%Also, for the aggregator, since we have $c_{\mc{A}}^{\star} = \sum_{j \in \mc{A}} c_j^{\star}$, and we showed that all the $c_j^{\star}$ are non-negative, so the $c_{\mc{A}}^{\star}$ is non-negative too. Another way to show this is to look at the closed form formula for the $c_{\mc{A}}^{\star}$ defined in \eqref{newsvendorFormula}. Since the $x_{\mc{A}}$ is always non-negative, $f_{\mc{A}} \left(x_{\mc{A}}\right)$ (and consequently $\mc{F}_{\mc{A}} \left(x_{\mc{A}}\right)$) is zero for negative values of $x_{\mc{A}}$. So the $c_{\mc{A}}^{\star}$ defined in \eqref{newsvendorFormula} would be always non-negative too.
%

\section{Proof of Theorem \ref{thmNE}} \label{thm1prf}

From the best responding condition of NE \eqref{defNE}, a necessary condition for a set of DA commitments $\{c_i^{ne}\}$ to be at a pure NE is 
\begin{align}   \label{eq:zeroderiv}
\frac{d  \pi_i}{d  c_i} \Big|_{\{c_i\} = \{c_i^{ne}\}}  =  0,  ~ \forall i=1,\ldots,N.  
\end{align}
Given the proposed PAM, with a little algebra, the above derivative can be expressed as 
		\begin{align}  
		&\frac{d  \pi_i}{d  c_i} = -\left(p^{r,b} - p^{r,s}\right)  c_i  f_{X_{\mc{N}}}  \left(x_{\mc{N}} =   \sum_{j \in \mc{N}} c_j\right)  \nonumber \\
		& \hspace{4pt} + \left(p^{r,b} - p^{r,s}\right)  \int_{0}^{\sum_{j \in \mc{N} } c_j} \hspace{-10pt}  x_i f_{X_{\mc{N}}, X_i} \left( x_{\mc{N}} = \sum_{j \in \mc{N}} c_j, x_i \right) d x_i  \nonumber \\
		& \hspace{4pt} + \left(p^f - p^{r,s}\right)  - \left(p^{r,b} - p^{r,s}\right) \mc{F}_{\mc{N}} \left( x_{\mc{N}} = \sum_{j \in \mc{N}} c_j \right). \label{dpi_i}
		\end{align}
With \eqref{dpi_i}, the sum of all the $N$ equations \eqref{eq:zeroderiv} simplifies as
		\begin{align}   \label{eq:sum}
		0 & = \sum_{i \in \mc{N}}  \frac{d  \pi_i}{d  c_i}\Big|_{\{c_i\} = \{c_i^{ne}\}}   \nonumber   \\
		& = N \times \left( \left(p^f - p^{r,s}\right) - \left(p^{r,b} - p^{r,s}\right) \mc{F}_{\mc{N}} \left( \sum_{j \in \mc{N}} c_j^{ne}  \right)  \right)
		\end{align}
From \eqref{eq:sum}, we obtain that 
\begin{align}
c_\mc{N}^{ne} = \sum_{j \in \mc{N}} c_j^{ne} = \mc{F}_{\mc{N}}^{-1} \left(\frac{p^f - p^{r,s}}{p^{r,b} - p^{r,s}}\right) = c_{\mc{N}}^{\star}, 
\end{align}
which in fact proves the \emph{efficiency} of the NE (cf. Corollary \ref{coroeff}). 

Now, substituting $\sum_{j \in \mc{N}} c_j= c_{\mc{N}}^{\star}$ into \eqref{dpi_i} and \eqref{eq:zeroderiv}, we get the unique solution 
\begin{align}
c_i^{ne} & =  \frac{\int_{0}^{c_{\mathcal{N}}^{\star}}  x_i f_{X_{\mc{N}}, X_i} \left( x_{\mc{N}} = c_{\mathcal{N}}^{\star}, x_i \right) d x_i}{  f_{X_{\mc{N}}}  \left(x_{\mc{N}} =  c_{\mathcal{N}}^{\star}\right)}  \nn\\
& = \mathbb{E} \left[ X_i  \Big| X_{\mathcal{N}} = c_{\mathcal{N}}^{\star} \right]. 
\end{align}

{ %\color{blue}
\section{Proof of Theorem \ref{conditionNE}}  \label{AppdxProofconditionNE}
From Theorem \ref{thmNE}, if a pure NE exists, it must take the form of \eqref{closedformNE}. In this proof, we show that
\begin{enumerate}
\item If condition \eqref{conditionNEeq} holds, then \eqref{closedformNE} is indeed a pure NE, and 
\item If condition \eqref{conditionNEeq} does not hold, then there exists a set of DA and RT prices such that a pure NE does not exist. 
\end{enumerate}

\emph{Part 1): } Suppose condition \eqref{conditionNEeq} holds. 

%We also make the following observations: $\forall i \in \mc{N},$
%\begin{align*}
%&\lim_{c_i \longrightarrow \infty} \pi_i \left(c_i , c_{-i}\right) = - \infty ,    \\
%&\lim_{c_i \longrightarrow -\infty} \pi_i \left(c_i , c_{-i}\right) = - \infty
%\end{align*}
From Theorem \ref{thmNE}, the unique candidate for a pure NE is given by $c_i^{ne}  =    \mathbb{E} \left[ X_i  \Big| X_{\mathcal{N}} = c_{\mathcal{N}}^{\star} \right],\forall i \in \mc{N}$ (cf. \eqref{closedformNE}). The expected payoff of RPP $i$ at this candidate pure NE is 
\begin{align}   \label{ExAntPayoffPAM_NE}
 &\pi_i \left(c_i^{ne} , \left\{c_{-i}^{ne}\right\}\right) =   p^{r,s} \cdot \mu_i  \ + \nonumber \\
 &\left(p^{r,b} - p^{r,s}\right) \cdot \hspace{-3pt}  \int_{x_{\mc{N}}=0}^{c_{\mc{N}}^{\star}} \hspace{-3pt} \mathbb{E}_{X_i | X_{\mc{N}}} \left[X_i | X_{\mc{N}} = x_{\mc{N}}\right] \cdot f_{X_{\mc{N}}} \left(x_{\mc{N}}\right) d x_{\mc{N}}
\end{align}
For the strategy profile $\left\{c_{i}^{ne}\right\}$ to indeed be a pure NE, we must also have: %$$
\begin{align}  \label{condition_appnx}
\forall i \in \mc{N}, \forall c_i \in \mathbb{R}, \hspace{5pt}  \pi_i \left(c_i^{ne} , \left\{c_{-i}^{ne}\right\}\right) \geq \pi_i \left(c_i , \left\{c_{-i}^{ne}\right\}\right), 
\end{align}
where $\pi_i \left(c_i , \left\{c_{-i}^{ne}\right\}\right)$ is the expected payoff of RPP $i$ if it chooses $c_i$ as its strategy (\textit{i.e.}, its DA-commitment), and the other RPPs choose the strategies $\left\{c_{-i}^{ne}\right\}$. $\pi_i \left(c_i , \left\{c_{-i}^{ne}\right\}\right)$ can then be expressed in closed form as follows:
\begin{align}  \label{ExAntPayoffPAM}
&\pi_i \left(c_i , \left\{c_{-i}^{ne}\right\}\right) = \nonumber \\
& p^{r,s} \cdot \mu_i + \left(p^f - p^{r,s}\right) \cdot c_i   + \left(p^{r,b} - p^{r,s}\right) \cdot \bigg\{  -c_i \cdot \mc{F}_{\mc{N}} \left(c_{\mc{N}}\right)  \nonumber  \\
&  \hspace{10pt} + \int_{x_{\mc{N}}=0}^{c_{\mc{N}}} \mathbb{E}_{X_i | X_{\mc{N}}} \left[X_i | X_{\mc{N}} = x_{\mc{N}}\right] \cdot f_{X_{\mc{N}}} \left(x_{\mc{N}}\right) d x_{\mc{N}}\bigg\} 
\end{align}
where $c_{\mc{N}} = c_i + \sum_{j\ne i} c_j^{ne}$. % $c_{\mc{N}} = c_i - c_i^{ne} + c_{\mc{N}}^{ne} $. 
Substituting \eqref{ExAntPayoffPAM_NE} and \eqref{ExAntPayoffPAM}, we have that \eqref{condition_appnx} is equivalent to
\begin{align}   \label{NE_necessary_simplified}
&\forall i \in \mc{N}, \forall c_i \in \mathbb{R}, \nonumber \\
& \int_{x_{\mc{N}}=c_{\mc{N}}^{\star}}^{c_{\mc{N}}}  \hspace{0pt}  \bigg\{\mathbb{E}_{X_i | X_{\mc{N}}} \left[X_i | X_{\mc{N}} = x_{\mc{N}}\right] - c_i\bigg\} \hspace{0pt} \cdot f_{X_{\mc{N}}} \left(x_{\mc{N}}\right) d x_{\mc{N}}  \nonumber  \\
& \hspace{190pt}  \leq 0.
\end{align}

From Theorem \ref{thmNE}, when $c_i = c_i^{ne}$, we have that $\mathbb{E}_{X_i | X_{\mc{N}}} \left[X_i | X_{\mc{N}} = x_{\mc{N}}\right] - c_i = 0$, and the left hand side (LHS) of \eqref{NE_necessary_simplified} equals to zero. %and thus \eqref{NE_necessary_simplified} holds. 
Now, from the condition \eqref{conditionNEeq}, 
\begin{itemize}
\item If $c_i > c_i^{ne}$, then $c_{\mc{N}} > c_{\mc{N}}^{\star}$, and $\mathbb{E}_{X_i | X_{\mc{N}}} \left[X_i | X_{\mc{N}} = x_{\mc{N}}\right] - c_i \le 0, \forall x_\mc{N} \in[c_{\mc{N}}^{\star}, c_{\mc{N}} ]$. Thus, \eqref{NE_necessary_simplified} holds. 
\item If $c_i < c_i^{ne}$, then $c_{\mc{N}} < c_{\mc{N}}^{\star}$, and $\mathbb{E}_{X_i | X_{\mc{N}}} \left[X_i | X_{\mc{N}} = x_{\mc{N}}\right] - c_i \ge 0, \forall x_\mc{N} \in[c_{\mc{N}}, c_{\mc{N}}^{\star}]$. Thus, \eqref{NE_necessary_simplified} holds. 
\end{itemize}

\emph{Part 2): } Suppose condition \eqref{conditionNEeq} does not hold. In other words, there exists an RPP $k$, for some $\alpha$, $\frac{d \mathbb{E}_{X_k | X_{\mc{N}}} \left[X_k | X_{\mc{N}} = \alpha\right]}{d \alpha} > 1$. 

Assuming continuity of $\mathbb{E}_{X_k | X_{\mc{N}}} \left[X_k | X_{\mc{N}} = \alpha\right]$ as a function of $\alpha$, there exists an interval $\left[D_1 , D_2\right]$, where $\frac{d \mathbb{E}_{X_k | X_{\mc{N}}} \left[x_k | X_{\mc{N}} = \alpha\right]}{d \alpha} > 1$ for all $\alpha\in[D_1,D_2]$. 

Now, under the mild technical condition that $\mc{F}_{\mc{N}}$ is invertible, we can always find a set of prices $\left\{p^f, p^{r,b}, p^{r,s}\right\}$ that satisfy 
\begin{align}  \label{pricesProof}
c_{\mc{N}}^{\star} = \mc{F}_{\mc{N}}^{\hspace{1pt} -1} \left(\frac{p^f - p^{r,s}}{p^{r,b} - p^{r,s}}\right) = D_1. 
\end{align}
We now examine, for this RPP $k$, any $c_k \in \left(c_k^{ne}, c_k^{ne} + D_2 - D_1 \right)$. Note that, for $c_k = c_k^{ne}$, the LHS of \eqref{NE_necessary_simplified} is zero; and for any $c_k \in \left(c_k^{ne}, c_k^{ne} + D_2 - D_1 \right)$, the LHS of \eqref{NE_necessary_simplified} is positive. Therefore, %for any $c_k \in \left(c_k^{ne}, c_k^{ne} + D_2 - D_1 \right)$, 
under this particular set of prices, the unique candidate of a pure NE is \emph{not} a pure NE, and thus a pure NE does not exist. 
}

\section{Deriving the Payoff Allocation Mechanism from a Competitive Equilibrium} \label{sec:market}
In this section, we show that the proposed PAM \eqref{NewPAM} can in fact be derived from computing the competitive equilibrium of a specially formulated market with transferrable payoff. %The idea based on competitive equilibrium shares similar insight with that in a prior work \cite{ZQRGP15} which derives a closed-form PAM in the core in the ``ex-ante'' sense. 
This also offers a proof of Theorem \ref{expostthm}. 

\subsection{Market with Transferrable Payoff} \label{MTP}
We first define the following market with transferrable payoff \cite{GameBook}: 
\begin{itemize}
	\item The RPPs, denoted by $\mathcal{N}$, are a finite set of $N$ agents. 
	\item There is one type of input goods --- power generation. 
	\item Each agent $i \in \mathcal{N}$ has an ``endowment'' in the amount of $x_i \in \mathbb{R}_+$ --- the realized power of RPP $i$. 
	\item Each agent $i \in \mathcal{N}$ has a continuous, nondecreasing, and concave ``production'' function $f_i: \mathbb{R}_+ \rightarrow \mathbb{R}$:
	\begin{align}
	f_i(x_i) = \mathcal{P}_i^{sep} = p^f c_i - p^{r,b} \left(c_i - x_i\right)_+ + p^{r,s} \left(x_i - c_i\right)_+. \label{prodfun}
	\end{align}
\end{itemize}
Since all the ``production'' functions $\{f_i\}$ produce the same type of transferrable output, i.e., monetary payoff, the above formulation precisely defines a market with transferrable payoff. 

Next, a coalitional game can be defined based on a market with transferrable payoff \cite{GameBook}. Specifically, for any coalition of a subset of RPPs $\mathcal{T}\subseteq\mc{N}$, define 
\begin{align} \label{CoalValueFunMarketGame}
v \left(\mc{T}\right) = &\max_{\{z_i\in \mathbb{R}_+, i \in \mc{T}\}} ~ \sum_{i \in \mc{T}}f_i(z_i) \\
 &s.t.~ \sum_{i \in \mc{T}}{z_i} = \sum_{i \in \mc{T}}{x_i}. \nn
\end{align}
In other words, $\{z_i, i\in\mc{T}\}$ denotes a \emph{redistribution} of the total realized power $\sum_{i\in\mc{T}} x_i$ among the members of $\mc{T}$. This $v(\mc{T})$ represents the \emph{maximum} total payoff that the members of $\mc{T}$ can achieve among all possible redistributions, computed according to $f_i$ defined in \eqref{prodfun}. 
The core of this coalitional game is also called the ``core of the market''. 

We now prove that this coalitional game is exactly the same as the coalitional game defined previously in \eqref{ValOfCoal}. 
\begin{lemma} \label{equgame}
The values of coalitions \eqref{CoalValueFunMarketGame} are the same as \eqref{ValOfCoal}.
\end{lemma}
\begin{proof}
Straightforwardly, $\eqref{ValOfCoal} \ge  \eqref{CoalValueFunMarketGame}$ because \eqref{ValOfCoal} is the maximum achievable payoff by the subset $\mc{T}$ after their aggregation. Next, we show that \eqref{ValOfCoal} can be achieved by \eqref{CoalValueFunMarketGame}, i.e., $\eqref{ValOfCoal} \le  \eqref{CoalValueFunMarketGame}$. 

%Similarly as in the proof of Theorem \ref{incorethm}, 
We define ${\mc{T}}^+ \triangleq \left\{ i \in {\mc{T}} ~|~ x_i - c_i \geq 0 \right\}$ and $\mc{T}^- \triangleq \left\{ i \in {\mc{T}} ~|~ x_i - c_i < 0 \right\}$. The intuition of a redistribution $\{z_i\}$ to achieve \eqref{ValOfCoal} is the following: We give as much of the \emph{excess} power of the RPPs in $\mc{T}^+$ as possible to the RPPs in $\mc{T}^-$ to offset their \emph{deficit} power. 

Specifically, 
if $x_{\mc{T}} - c_{\mc{T}} < 0$, i.e., $\sum_{i\in\mc{T}^-} \left(c_i - x_i\right) > \sum_{i\in\mc{T}^+} \left(x_i - c_i\right)$, we let 
\begin{align}
\forall i\in\mc{T}^+,& ~z_i = c_i, \label{zplus}\\
\forall i\in\mc{T}^-,& ~ x_i\le z_i \le c_i, \nn\\
&\mbox{ so that } \sum_{i\in\mc{T}^-} \left(z_i - x_i\right) = \sum_{i\in\mc{T}^+} \left(x_i - z_i\right). \label{zminus}
\end{align}
As a result, 
\begin{align}
\sum_{i \in \mc{T}}f_i(z_i) &= \sum_{i \in \mc{T}^+}f_i(z_i) + \sum_{i \in \mc{T}^-}f_i(z_i) \nn\\
&= \sum_{i \in \mc{T}^+}p^f c_i + \sum_{i \in \mc{T}^-}\left(p^fc_i - p^{r,b}(c_i-z_i) \right) \nn\\
& =p^f c_{\mc{T}} - p^{r,b}\sum_{i \in \mc{T}^-}\left((c_i-x_i) - (z_i - x_i) \right) \nn\\
& =p^f c_{\mc{T}} - p^{r,b}\left(\sum_{i \in \mc{T}^-}(c_i-x_i) - \sum_{i \in \mc{T}^+}(x_i - c_i) \right) \label{2ndtolast}\\
& = p^f c_{\mc{T}} - p^{r,b}(c_{\mc{T}} - x_{\mc{T}}) =  \eqref{ValOfCoal}, 
\end{align}
where \eqref{2ndtolast} is implied by \eqref{zplus} and \eqref{zminus}. 

The case of $x_{\mc{T}} - c_{\mc{T}} \ge 0$ can be proved similarly. 
\end{proof}

As a result, from the property of market with transferrable payoff (cf. Proposition 264.2 in \cite{GameBook}), we immediately have that this coaltional game has a \emph{non-empty core}. 

Moreover, this formulation as a market enables us to compute a solution in the core by deriving the \emph{competitive equilibrium} (CE) of the market, as follows. 

\subsection{Competitive Equilibrium} 
For the market with transferrable payoff defined in the last subsection, a competitive equilibrium is defined \cite{GameBook} as a price-quantity pair of $p^*\in\mbb{R}_+$ and $\bm{z}^*\in\mbb{R}_+^N$, such that, 
\begin{itemize}
\item[i)]
For each agent $i$, $z_i^*$ solves the following problem: 
\begin{align}  \label{OptZi}
\max_{z_i \in \mathbb{R}_+}{\left( f_i \left(z_i\right) - p^* \left(z_i - x_i\right) \right)}. 
\end{align}
\item[ii)] $\bm{z}^*$ is a redistribution, i.e., $\sum_{i\in\mc{N}}z_i^* = \sum_{i\in\mc{N}}x_i$. 
\end{itemize}
The intuition of a CE is the following: At the price $p^*$, i) to maximize its payoff, each agent $i$ can trade \emph{any} amount of the input (realized power) on the market \emph{without} worrying whether there is enough supply or demand to fulfill its trade request, and ii) collectively, the market  of input supply and demand \emph{still clears}, i.e., the resulting $\bm{z}^*$ from the optimal trades is feasible. 

At a competitive equilibrium $(p^*, \bm{z}^*)$, $p^*$ is called the \emph{competitive price}, and the value of the maximum of \eqref{OptZi} is called the \emph{competitive payoff} of agent $i$. 

We then have the following theorem (cf. Proposition 267.1 in \cite{GameBook}) dictating that all the CEs are in the core.
\begin{theorem} \label{thmCEcore}
Every profile of competitive payoffs in a market with transferable payoff is in the core of the market. 
\end{theorem}

Accordingly, to find a solution in the core of the market, which is also the core of the coalitional game for aggregating RPPs (cf. Lemma \ref{equgame}), it is sufficient to find a CE in the market with transferrable payoff defined in the last section. 
\vspace{5pt}

\noindent {\bf Deriving the Competitive Equilibrium:}

For the market with transferrable payoff defined in the last subsection, we have the following theorem: 
\begin{theorem} \label{thmCE}
Competitive equilibrium exists, and the competitive payoffs necessarily take the form of the proposed PAM \eqref{NewPAM}. 
\end{theorem}
\begin{proof}
With the production function $f_i(x_i)$ defined to be $\mathcal{P}_i^{sep}$ as in \eqref{prodfun}, we observe that $f_i(x_i)$ is a \emph{piecewise linear} function: $f_i'(x_i) = 
\begin{cases}
p^{r,b}, \mbox{ if } x_i < c_i \\
p^{r,s}, \mbox{ if } x_i > c_i
\end{cases}\!\!\!\!. $

As a result, at a CE, we must have $p^{r,b} \le p^* \le p^{r,s}$. Otherwise, by solving \eqref{OptZi}, either all RPPs would sell all of their power, or all of them would buy an infinite amount of power; Neither case would clear the market with $\sum_{i\in\mc{N}}z_i^* = \sum_{i\in\mc{N}}x_i$. 

We now analyze the optimal behavior of any agent $i$ under the following three scenarios of the competitive price $p^*$: 
\begin{itemize}
\item If $p^* = p^{r,b}$, 
the maximum of \eqref{OptZi} is achieved if and only if $z_i \le c_i$. 
\item If $p^* = p^{r,s}$, 
the maximum of \eqref{OptZi} is achieved if and only if $z_i \ge c_i$. 
\item If $p^{r,s} < p^* < p^{r,b}$,  
the maximum of \eqref{OptZi} is achieved if and only if $z_i = c_i$. 
\end{itemize}

To derive the competitive price $p^*$ that clears the market with $\sum_{i\in\mc{N}}z_i^* = \sum_{i\in\mc{N}}x_i$, we consider the following three scenarios: 

\emph{Case i)} $x_\mc{N} - c_\mc{N}<0$: 
As result, at the CE, $\sum_{i\in\mc{N}}z_i^* < \sum_{i\in\mc{N}}c_i^*$. From the above, we \emph{necessarily} have $p^* = p^{r,b}$. Indeed, with $p^* = p^{r,b}$, there exists $\bm{z}^*$ such that a) $z^*_i \le c_i$, and b) $\sum_{i\in\mc{N}}z_i^* = \sum_{i\in\mc{N}}x_i < c_\mc{N}$. 

Moreover, it is immediate to check that the competitive payoff of RPP $i$ equals $p^f c_i +  		p^{r,b} \left(x_i - c_i\right)$ (cf. \eqref{NewPAM}). 

\emph{Case ii)} $x_\mc{N} - c_\mc{N}>0$: 
As result, at the CE, $\sum_{i\in\mc{N}}z_i^* > \sum_{i\in\mc{N}}c_i^*$. From the above, we \emph{necessarily} have $p^* = p^{r,s}$. Indeed, with $p^* = p^{r,s}$, there exists $\bm{z}^*$ such that a) $z^*_i \ge c_i$, and b) $\sum_{i\in\mc{N}}z_i^* = \sum_{i\in\mc{N}}x_i > c_\mc{N}$. 

Moreover, the competitive payoff of RPP $i$ equals $p^f c_i +  		p^{r,s} \left(x_i - c_i\right)$ (cf. \eqref{NewPAM}). 

\emph{Case iii)} $x_\mc{N} - c_\mc{N}=0$: 
In this case, $\forall p^*,~s.t.~p^{r,s} \le p^* \le p^{r,b}$, $z_i^* = c_i, \forall i$ achieves $\sum_{i\in\mc{N}}z_i^* = \sum_{i\in\mc{N}}x_i = c_\mc{N}$. 

Moreover,  the competitive payoff of RPP $i$ equals $p^f c_i +  		p^* \left(x_i - c_i\right)$ (cf. \eqref{NewPAM}).
\end{proof}

From Theorem \ref{thmCEcore} and \ref{thmCE}, we conclude that the competitive payoffs that equal \eqref{NewPAM} are always in the core of the market, and hence the core of the coalitional game \eqref{ValOfCoal}.

\end{document}